\documentclass[journal,onecolumn, 12pt]{IEEEtran}
\usepackage{algorithmic,algorithm}

   \usepackage{setspace}
\usepackage{graphicx,mathptmx,amsmath,amsfonts,amsthm,color,hyperref,comment}
\usepackage{enumitem}
\usepackage{enumerate}
\usepackage{xcolor}

\newtheorem{assumption}{Assumption}
\newtheorem{theorem}{Theorem}
\newtheorem*{example*}{Example}

\newtheorem{proposition}{Proposition}

\def  \MYEQ{\stackrel{\text{?}}{=}}

\def\ben{\begin{eqnarray*}}
\def\een{\end{eqnarray*}}

\def  \logA{\log_{|\A|}}
\def  \logT{\log_{2}}
\def  \MYEQ{\stackrel{\text{?}}{=}}

\def  \F{\mathbb {F}}
\def  \F2{\mathbb {F}_2}
\def  \A{\mathbb {A}}

\def  \R{\mathbb {R}}

\def  \E{\mathbb {E}}

\def \cC{\mathcal{C}}

\def \logA{\log}

\def\GN{ G }
\def\HN{ H }
\def\Hmin{ H_{\min} }
\def\Hhalf{ H_{1/2} }

\def\LambdaN{ \Lambda^{N} }
\def\IN{ I^{N} }
\def\INLjoint{ I^{N,L} }
\def\INL{ I^{{N}^L} }
\def\LambdaNL{ {\Lambda^{{N}^L}} }

\def\IL{ I^L }
\def\LambdaL{ \Lambda^L }
\def\muL{ \mu^L }
\def\Noise{ N }
\def\UNoise{ N_1 }

\def\IU{ I^U }

\def\h2{ h_2 }

\def\gstar{ g^* }

\def\error{ \epsilon }
\def\errorAML{ \epsilon^\text{AB} }

\def\errorL{ \epsilon^L }

\def\CHARD{ C^\text{Hard} }
\def\CSOFT{ C^\text{Sym. Reliability} }
\def\CSOFT{ C^\text{Sym. Rel.} }

\def\GML{ D }
\def\IMLE{I^\text{D}}
\def\IMLEAB{I^\text{D-AB}}

\def\GAML{ D_\text{AB} }

\def\LC{ C }

\def  \codebook{\Phi}

\def  \MERR{ \text{MER} }
\def  \NORMAL{ \mathcal{N} }
\def  \Ninv{ F_{\NORMAL}^{-1} }

\begin{document}

\title{Guessing random additive noise decoding with  symbol reliability information (SRGRAND)}


\author{
\IEEEauthorblockN{Ken R. Duffy\IEEEauthorrefmark{1}, Muriel M\'edard\IEEEauthorrefmark{2} and Wei An\IEEEauthorrefmark{2}\\}
\IEEEauthorblockA{\IEEEauthorrefmark{1}Hamilton Institute, Maynooth University, Ireland. E-mail: ken.duffy@mu.ie.\\}
\IEEEauthorblockA{\IEEEauthorrefmark{2}Research Laboratory of Electronics, Massachusetts Institute of Technology, 
U.S.A.
E-mail: medard@mit.edu, wei\_an@mit.edu.}
}

\maketitle

\vspace{-0.9in}


\begin{abstract}
The design and implementation of error correcting codes has long been informed by two fundamental results: Shannon's 1948 capacity theorem, which established that long codes use noisy channels most efficiently; and Berlekamp, McEliece, and Van Tilborg's 1978 theorem on the NP-hardness of decoding linear codes. These results shifted focus away from creating code-independent decoders, but recent low-latency communication applications necessitate relatively short codes, providing motivation to reconsider the development of universal decoders. 

We introduce a scheme for employing binarized symbol soft information within Guessing Random Additive Noise Decoding, a universal hard detection decoder. We  incorporate codebook-independent quantization of soft information to indicate demodulated symbols to be reliable or unreliable. We introduce two decoding algorithms: one identifies a conditional Maximum Likelihood (ML) decoding; the other either reports a conditional ML decoding or an error. For random codebooks, we present error exponents and asymptotic complexity, and show benefits over hard detection.

As empirical illustrations, we compare performance with majority logic decoding of Reed-Muller codes, with Berlekamp-Massey decoding of Bose-Chaudhuri-Hocquenghem codes, with CA-SCL decoding of CA-Polar codes, and establish the performance of Random Linear Codes, which require a universal decoder and offer a broader palette of code sizes and rates than traditional codes.

\end{abstract}

\vspace{-0.2in}
\section{Introduction}

\let\thefootnote\relax\footnotetext{
A subset of these results was presented at the 2019 IEEE International
Symposium on Information Theory, Paris, France, \cite{Duffy19a} and
at the 2020 Annual Conference on Information Sciences and Systems,
Princeton, USA \cite{Duffy20}.
In this article lower case letters
correspond to realizations of upper-case random variables or their
normalized limits, apart from for noise where $z$ is used as $n$
denotes the code block-length. Logs are taken base $|\A|$ throughout,
and we assume that $0\in\A$ corresponds to no noise.}

Since Shannon's 1948 opus \cite{Shannon48} it has been known that
channel capacity, the highest rate that an error correcting code
can operate at while guaranteeing error-free communication over a
noisy channel, is governed by the Shannon entropy of the channel's
noise. By considering structureless random codes, his mathematical results proved that
channel capacity is only achievable in the limit as the length of
the error correcting code becomes large. By 1968, it was
confirmed that his core theorems hold if structureless random codes
are replaced with Random Linear Codes (RLCs)~\cite{Gal68}, which
offer a more efficiently stored codebook description. In 1978,
however, Berlekamp, McEliece, and Van Tilborg reported that maximum likelihood (ML) 
decoding of linear codes is an NP-complete
problem~\cite{berlekamp1978inherent}, establishing that there exists
a sequence of linear codes for which the decoding complexity is
exponential as a function of block length. This feature, which
underpins the McEliece cryptosystem \cite{mceliece1978public},
effectively halted practical consideration of universal decoding
algorithms, with a couple of notable exceptions recounted in the
Related Work.

The focus on long codes led to a working paradigm of  pairing structured codes and  code-specific decoders.  Examples of such pairings are Reed-Muller (RM)
codes~\cite{reed1954class,muller1954application} with Majority Logic
decoding, Reed-Solomon codes~\cite{reed1960polynomial} with
Berlekamp-Massey (BM) decoding \cite{berlekamp1968algebraic,massey1969shift,shu2004},
Low Density Parity Check Codes (LDPCs) \cite{gallager1963low} with
belief propagation decoding~\cite{fossorier1999reduced}, and, most
recently, CRC-Assisted Polar (CA-Polar) codes,  used
in control channel communications in 5G New Radio (NR), with CRC-Assisted
Successive Cancellation List (CA-SCL) decoding
~\cite{niu2012crc,tal2015list,balatsoukas2015llr}.  The
structured nature of these codes leads to restrictions
on lengths and rates. They are usually constructed based  on the
assumption of independent and identically
distributed noise, which is then approximated through
significant interleaving, with attendant delays. From
an implementation point of view,
distinct hardware is required for each code-decoder
pair, and sometimes for different rates of the same code-decoder
pair.

Many current communication systems 
require
low-latency operation, where small bursts of data need 
efficient transmission \cite{durisi2016toward,she2017radio,chen2018ultra}.
Indeed, ultra-reliable low-latency communication (URLLC) is an
important use-case in the 5G NR standard
\cite{parvez2018survey,medard20205}. Delivering URLLC necessitates
efficient decoding of short, high-rate codes,  motivating revisiting the possibility of high-accuracy universal
decoders. Guessing Random Additive Noise Decoding
(GRAND)~\cite{Duffy18,Duffy19,An20}, first proposed in 2018 for
hard detection channels, is a class of decoding algorithms that can
decode any code. GRAND's practical promise as a single efficient
mechanism for any moderate redundancy code is
such that circuit-based implementations have already been investigated
\cite{abbas2020high,Riaz21,abbas2021b} that avail of the inherent high level of
parallelizability of the algorithm. That work demonstrates GRAND's
performance credentials in hard detection channels, such as  data storage system applications or communication systems
with only hard detection demodulation.

GRAND's universal premise is that for a communication to be decodable
the received signal must faithfully contain information regarding
the transmitted code-word \emph{and} the error effect of the noise
experienced on the channel. While most decoding algorithms
utilize the codebook's structure to identify the transmitted
code-word, GRAND endeavors to find the effect of the noise 
and so recover the transmitted code-word.
To do this, it requires two devices: a method by which to query if
a string is an element of the codebook; 
and a mechanism to sequentially create putative noise-effect
sequences in decreasing order of their likelihood of occurrence on
the channel. Armed with these, GRAND aims to produce an error corrected
decoding for \emph{any} block code, without restriction to binary
or, indeed, linear codes.

\begin{algorithm}
\caption{Guessing Random Additive Noise Decoding. Given a demodulated
channel output $y^n$ and a function $\codebook$ such that $\codebook(y^n)=0$ if and only if
$y^n$ is in the codebook, $c^{n,*}$ is the first codebook element 
identified and $\GML^n$ is the number of codebook queries required to
identify it, serving as a measure of confidence.}
\label{alg:pseudo-code}
\begin{algorithmic}
\STATE {\bf Inputs}: $y^n$, $\codebook$
\STATE {\bf Output}: $c^{n,*}$, $\GML^n$
\STATE $d\leftarrow 0$, $\GML^n\leftarrow 0$.
\STATE $z^n\leftarrow$ next most likely noise effect sequence
\STATE $\GML^n\leftarrow\GML^n+1$ 
 \IF{$\codebook(y^n\ominus z^n) = 0$}
\STATE  $c^{n,*}\leftarrow y^n\ominus z^n$
\STATE{\bf return} $c^{n,*}$,  $\GML^n$
\ENDIF
\end{algorithmic}
\hrule
\end{algorithm}

Pseudo-code for GRAND can be found in Algorithm ~\ref{alg:pseudo-code},
where the key step is ``$z^n\leftarrow$ next most likely noise
effect sequence''. In the work that introduced GRAND the decoder only had access to a statistical description of
the channel and hard-detection information. In that setting, GRAND provides ML decoding so long as the
ordering of the putative noise effects matches the statistical
description of the channel, even for channels with temporal noise
correlations \cite{Duffy19}. For standard models of hard detection noise effects from a highly 
interleaved channel or one subject to Markovian burst errors 
without an interleaver, dynamically creating putative noise effects
is possible with simple logic \cite{An20}. That putative noise effects 
can be readily generated in parallel has been 
exploited in published hardware implementations of GRAND that perform multiple codebook
membership queries per clock-cycle \cite{abbas2020high,Riaz21,abbas2021b}. 

Incorporating soft information from per-realization
measurements of received signals is  known to be able to improve decoding significantly
 \cite{goldsmith2005wireless}, but it is unclear how to do
so  with GRAND. It seems fraught at first blush as soft information
seeks to represent continuous observations at the receiver, while
GRAND searches over a collection of discrete noise effects on
demodulated signals. A na\"ive approach, in which fine quantization
of noise leads to guessing over a larger space of possible noise
realizations, is inherently undesirable from a complexity perspective.

Here we consider the problem of incorporating binary symbol reliability information to GRAND where soft information per received symbol is limited to a single bit to indicate whether a demodulated symbol has been demodulated with confidence or not. Analysis of the resulting schema, Symbol Reliability GRAND (SRGAND), results in: a universal ML decoder 
conditioned on one bit of symbol reliability information per received symbol; simulated performance for established linear codes and for RLCs, 
which exist at all lengths and rates and have theoretically desirable properties \cite{coffey1990any}, but require a universal decoder; the mathematical evaluation of SRGRAND's complexity, showing an improvement vis-\`a-vis GRAND; error exponents for conditional ML decoding in the presence of a single bit of symbol reliability, and success exponents for the likelihood of correct decoding when the code-rate exceeds capacity.

\section{Symbol reliability.}
Essentially all digital communications involve taking discrete
data, channel coding them to add robustness to noise, and then
modulating those digital data into signals suitable for
transmission and reception. For example, Phase Shift Keying (PSK), widely used in wireless
communications systems, encodes groups of binary data into one
of a finite set of phases of a  carrier signal. To avail of better channel
conditions in practice, not only is the codebook rate increased,
but a modulation with a larger number of bits per modulated
symbol is also employed. An illustration of Quadrature PSK (QPSK) is provided in Fig. \ref{fig:0}. With transmitted symbols indicated by the red
dots, assuming all symbol transmissions are equally likely and they are disturbed by independent additive Gaussian channel noise (AWGN), the probability density of a received signal being observed is indicated by the heat maps in Fig. \ref{fig:0}
(a). Hard detection demodulation maps each received signal to the
nearest potentially-transmitted symbol. One metric of
confidence that a hard demodulated symbol corresponds to the transmitted
one is the minimum across all possible alternate symbols of the Likelihood Ratio (LR) that a received signal was
observed given the hard detection symbol was transmitted as compared
with the alternate. The resulting LR surface
is depicted in Fig. \ref{fig:0} (b). Instead of solely reporting the hard detection output, we envisage a further codebook independent quantization of the received signal
into a symbol reliability indicator that separates reliably received
symbols from unreliable ones. The principle behind the approach is
illustrated in Fig. \ref{fig:0} (c) where a thresholding of the LR
results in a masked region such that if a signal is received within
that region, the hard detection demodulated symbol is flagged as
being unreliable. The probability density of receiving a signal
conditional on being in the masked region of 
 potentially noise-impacted symbols is shown in Fig. \ref{fig:0}
(d). Thus, the uncertainty region corresponding to the
potentially noise-impacted symbols serves as a mask that labels
received symbols whose values are questionable, enabling the decoder
to focus on them. As the GRAND approach is codebook independent and
noise-centric, we establish that we can incorporate  symbol
reliability information in a way that results in reduced complexity.

\begin{figure*}
\vspace{-0.2in}
\begin{center}
\includegraphics[width=0.24\textwidth]{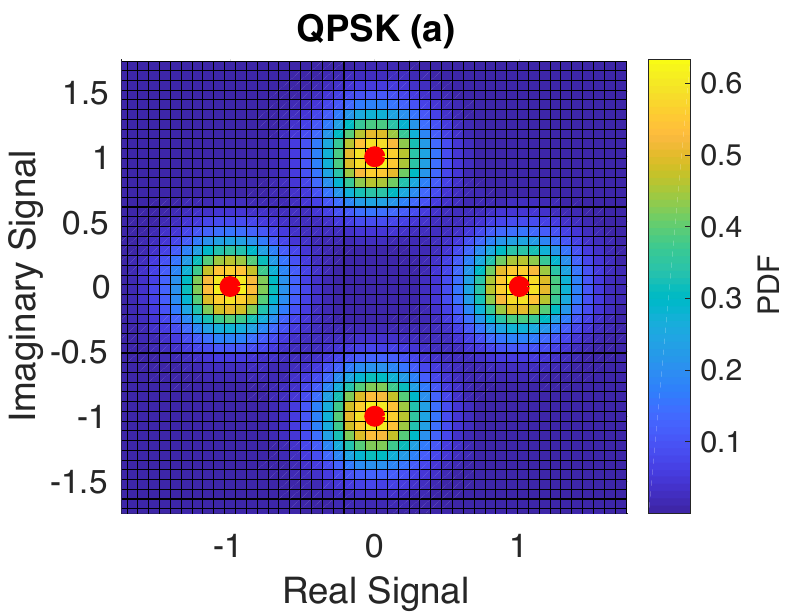}
\includegraphics[width=0.24\textwidth]{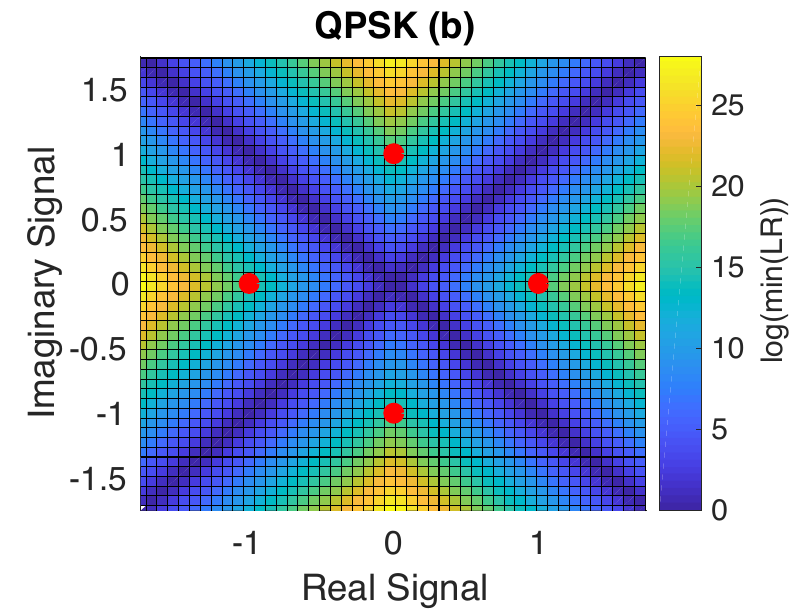}
\includegraphics[width=0.24\textwidth]{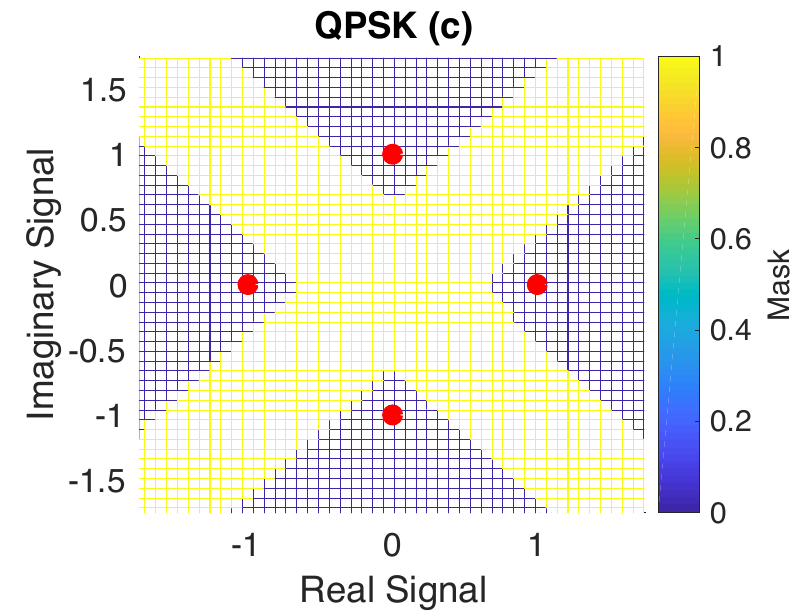}
\includegraphics[width=0.24\textwidth]{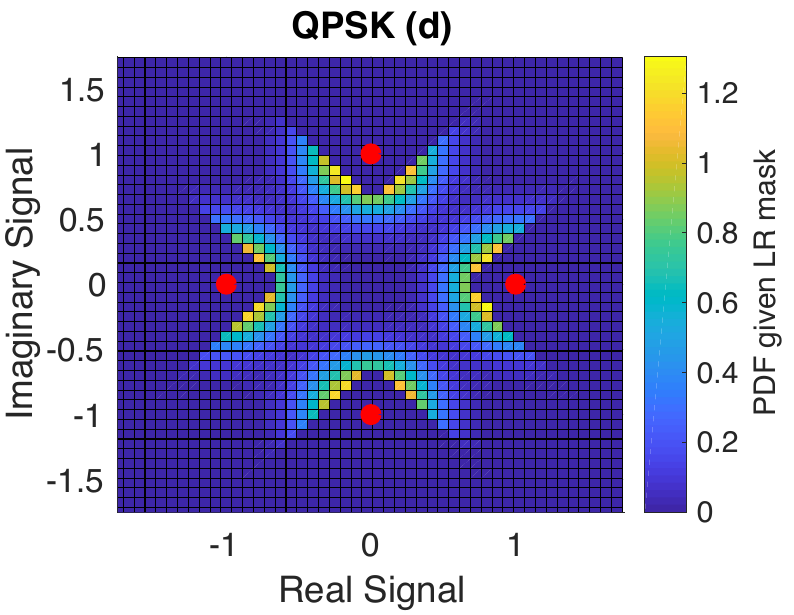}
\vspace{-0.2in}
\end{center}
\caption{QPSK subject to uncorrelated bivariate AWGN. Each pair
of bits is coded into one of four symbols, indicated by the red
dots. (a) displays a heat map of the probability
density that the received signal is at a given location. When hard
detection is employed, received signals are demodulated to the
symbol in the quadrant where the received signal is
measured and that is provided to the decoder. (b)
shows the  minimum Likelihood Ratio (LR) between each hard detection
symbol and all others as heat maps, providing a measure of
confidence in the hard demodulation. (c) displays a
mask of the LR surfaces in (b) where within the hatched area the LR is
greater than a threshold, and in the yellow mask area it is less than
the threshold, identifying the potentially noise-impacted symbols. In the symbol reliability quantization, symbols received in the yellow masked region are
demodulated but are also marked as being potentially noise-impacted. (d)  provides heat map views of the probability density function
of a received signal, conditioned on it being observed in the yellow masked
area of uncertainty.\vspace{-0.2in}}
\label{fig:0}
\end{figure*}

Our mathematical abstraction of this reliability information assumes that symbols received from the channel have been accurately indicated to be error free or to have possibly been
subjected to independent additive random noise. SRGRAND then provides an ML decoding conditioned on the mask that was provided. In practice, creation of this symbol reliability information corresponds to a situation where soft information, such as instantaneous Signal to Interference plus Noise Ratios (SINR), has been thresholded so as to provide false negatives with a sufficient small likelihood that poor masking, i.e. incorrect identification of potentially noise-impacted symbols, does not dominate the block error probability. In effect, this symbol reliability information is a codebook-independent quantization of soft information \cite{WS05}. In Section \ref{sec:perf}, using a simple threshold rule, for an additive white Gaussian noise channel we empirically find that the provision of symbol reliability information results in a 0.75 to 1 dB gain over optimal ML hard detection decoding, even when the symbol reliability information is potentially erroneous. 

\section{Guessing Random Additive Noise Decoding}
\label{sec:GRAND}

The contribution of the current article is to identify how to
incorporate symbol reliability information into the
GRAND approach, and the ensuant increased capacity, reduced block error
probability, and decreased complexity.
 We assume that, as well as being in receipt of
a channel output, $Y^n$, the receiver is provided with a vector of
symbol reliability information, $S^n$ taking values
in $\{0,1\}^n$ where a $0$ truthfully indicates a symbol has not
been subject to noise while a $1$ indicates it may have been. This
model is similar in spirit to the well-known Gilbert-Elliott model
\cite{Gilbert1960, Elliott1963}, although our results will hold for
channel state process $\{S^n\}$ that have more involved correlation
structures than Markovian. The core idea is that the vector $S^n$
be used as a mask that separates symbols that require guessing, since they are potentially noise-impacted, from
those that do not. 

\subsection{GRAND}

Consider a hard-detection channel with inputs, $X^n$, and outputs, $Y^n$, consisting
of blocks of $n$ symbols from a finite alphabet $\A=\{0,\ldots,|\A|-1\}$.
Assume that channel input is altered by random noise effects, $\Noise^n$, that are
independent of the channel input and also take values in $\A^n$.  
Assume  the function, $\oplus$, describing the channel's action, is invertible so that knowing the output and input the noise can be recovered:
\begin{align}
    \label{eq:channel}
    Y^n = X^n \oplus \Noise^n \text{ and } \Noise^n=Y^n \ominus X^n.
\end{align}
In the hard detection setting, the receiver
is solely provided with the discrete channel output $Y^n$.

Assuming
code-words are selected uniformly at random, to implement
ML decoding, the sender and receiver first
share a codebook $\cC_n=\{c^{n,1},\ldots,c^{n,M_n}\}$ consisting
of $M_n$ elements of $\A^n$. For a given channel output $y^n$,
denote the conditional probability of the received sequence given
the transmitted code-word was $c^{n,i}$ by $p_{Y^n|C^n}(y^n|c^{n,i})
= P(\Noise^n=y^n\ominus c^{n,i})$ for  $i\in\{1,\ldots,M_n\}$.  The
ML decoding is then 
\begin{align}
    \label{eq:straight_MLE}
    c^{n,*}\in\arg\max \left\{ p_{Y^n|C^n}(y^n|c^{n,i}): c^{n,i}\in\cC_n\right\}.
\end{align}

For hard detection, the principle underlying the algorithms in
\cite{Duffy18,Duffy19} is to focus on identifying the noise that
was experienced in the channel rather than directly trying to
identify the transmitted code-word.
Based on a statistical model of the symbol-level channel,
the receiver achieves this
by first rank-ordering noise sequences from most likely to least
likely, breaking ties arbitrarily. In that order, the decoder
sequentially queries whether the sequence that remains when the effect of the
putative noise is removed from the received signal is an element
of the codebook. The first instance where the answer is in the
affirmative is the decoded element. To see that 
GRAND corresponds to ML
decoding for channels described in Eq. \eqref{eq:channel}, note that, owing to the definition of $c^{n,*}$ in
Eq. \eqref{eq:straight_MLE},
\begin{align*}
   p_{Y^n|C^n}(y^n|c^{n,*})=
	P(\Noise^n=y^n\ominus c^{n,*})\geq P(\Noise^n=y^n\ominus
	c^{n,i}) \text{ for all } c^{n,i}\in\cC_n.
\end{align*} 
 Irrespective of how the codebook is
constructed, by sequentially subtracting noise sequence effects from the received
sequence in order from the most likely to least likely and querying
if it is in the codebook, the first identified element is a ML
decoding. GRAND can be thought of as a guessing race where
the querying process is halted either with success on identifying
the true noise, and hence the transmitted code-word, or with an
error on identifying a non-transmitted element of the codebook
\cite{Duffy19}. The second algorithm considered in \cite{Duffy19},
GRANDAB (GRAND with ABandonment), follows the same procedure as
GRAND, but abandons noise guessing and declares an error if more
than $|\A|^{n(H+\delta)}$ queries have been made, where $H$ is the
Shannon entropy rate of the noise and $\delta>0$ is arbitrary. If
more than $|\A|^{n(H+\delta)}$ queries are needed to identify a ML
decoding, then the noise has been sufficiently unusual that, in
query number terms, it is beyond the Shannon typical set. As a
result, the block-error rate cost of abandoning is asymptotically
negligible. Note the conditional likelihood that a ML
decoding is in error increases as the number of queries made before
identification of a codebook element increases, so one is abandoning a less certain decoding.


\subsection{SRGRAND}

The adaptation of this noise guessing principle to the symbol reliability
setting results in a ML decoder conditioned on the veracity of that 
symbol reliability information. SRGRAND that proceeds as follows:
\begin{itemize}
\item Given channel output $y^n$ and symbol reliability information
$s^n=(s^n_1,s^n_2,\ldots,s^n_n)$, initialize $i=1$, set the non-noise-impacted symbol
locations of guessed noise sequence $z^n$ to $0$,
and set the masked
(i.e. potentially noise-impacted) entries of $z^n$ to be the most likely noise effect sequence of
length $l^n=\sum_i s^n_i$.
\item While $x^n=y^n\ominus z^n\notin\cC_n$, increase $i$ by $1$
and change the masked potentially noise-impacted symbols 
$z^n$ to be the next most likely noise effect sequence
of length $l^n$.
\item The $x^n$ that results from this while loop is the decoded element.
\end{itemize}
Note that SRGRAND can directly co-opt sequential noise pattern generators 
that were developed for GRAND
by restricting their application to masked symbols alone.

Based on the same logic as for GRAND, which has only hard detection information,
this procedure identifies a conditional ML decoding in this setting, but,
depending on $s^n$, it will have performed fewer queries than GRAND and the
output element is more likely to be the transmitted one, owing to
the targeted nature of the querying.  While SRGRAND always returns
an element of the codebook that is a ML decoding conditioned on the
symbol reliability information, the version with
abandonment, SRGRANDAB, either provides a conditional ML decoding or returns an
erasure.  Without impacting
the capacity-achieving nature of the decoder, several distinct abandonment thresholds, which can be used
in combination, are possible and result in reduced decoding complexity. We comment on two other possibilities in
Section \ref{sec:discussion}, and prove results for one representative
rule:
\begin{itemize}
\item
With $L^n =\sum_{i=1}^n S^n_i$ being the random number of potentially
noise-impacted symbols, assuming it exists, let $\lim_n E(L^n/n) = \muL>0$ be the
long run average proportion of potentially noise-impacted symbols.
SRGRANDAB proceeds as SRGRAND, but abandons and declares an error
without providing an element of the codebook if more than
$|\A|^{n(\muL\HN+\delta)}$ queries are made, where $\HN$ is the
Shannon entropy of the noise for a potentially noise-impacted symbol,
and $\delta>0$ is arbitrary.
\end{itemize} 
This is similar to the GRANDAB abandonment rule, but where enough
queries are made to cover the typical set of the average number of
potentially noise-impacted symbols.

In Section \ref{sec:math} we mathematically determine the gain in capacity, reduction in block error
rate, and decrease in complexity that can be obtained by leveraging
this symbol reliability information within the GRAND
approach. The desirable features of GRAND stem from its focus on
the noise rather than on the codebook as transmissions that are
subject to light noise are quickly decoded, irrespective of the
codebook construction or its rate, and these properties are
preserved as we incorporate the symbol reliability information.
We illustrate the gains to be obtained by considering
a worked mathematical example in Section \ref{sec:SR-BSC} and, in Section \ref{sec:perf}, simulated performance evaluation with Reed-Muller (RM), Bose-Chaudhuri-Hocquenghem (BCH), CA-Polar, and RLC that also treats the possibility of decoding errors due to erroneous masks.

\section{Related Work}
\label{sec:related}

While the vast majority of codes and decoding systems, including
all those currently used in practice, are co-designed, a few universal
decoders have been developed.
The original ML decoder works by computing
the conditional probability of the received signal for every element
of the not-necessarily-linear codebook and selecting the most likely.
This approach means it can be used with structureless codes
stored in a dictionary, and for channels with memory
so long as the decoder has an accurate statistical description of
it. This brute force evaluation requires an enormous number of
real-valued computations for every received code-word, rendering
the approach infeasible for all but the shortest of codes
\cite{Alessio20}. It is, however, amenable to mathematical analysis
and remains of theoretical importance in the provision of performance
bounds for an optimal decoder.

Restricting to binary linear $[n,k]$ codes, universal decoders have
been studied for both cryptographical and communications purposes.
Finding its roots in Prange's seminal research \cite{prange1962use},
Information Set Decoding (ISD) and its variants
\cite{stern1988method,lee1988observation,leon1988probabilistic,peters2010information,bernstein2011smaller,becker2012decoding}
are randomized algorithms used to assess mathematically 
the security provided by code-based cryptosystems as the code becomes
long. The core cryptographic scenario essentially maps to memoryless
hard detection channels. Given a binary linear code-word
and a received hard detection communication that has been subject
to a known number of flipped bits, for each demodulated binary output
the basic version of ISD works in two iterated steps until a decoding
is found. The first is a transformation where the columns of the
binary parity check matrix are randomly permuted and Gaussian
elimination is performed to rewrite the code in systematic format.
In the second step, for a number of columns that is less than the
code's correction capability, the difference between the syndrome
and all linear combinations of that number of columns is evaluated.
Once this difference is found to be the zero vector, the Gaussian
elimination transformation is inverted to identify the decoded
code-word. Probabilistic analysis of the algorithm provides worse
case bounds for decoding any code, and later tweaks to the algorithm
serve to reduce the exponent in the complexity as a function of
code-length. To use ISD for communications requires some adaptation
owing to the assumption of an {\it a priori} known number of flipped
bits. 

In communications, soft information has been exploited to produce
approximate ML decoders for binary linear codes.  In 1974
\cite{dorsch1974decoding} Dorsch introduced the idea of the Most
Reliable Basis (MRB), and developments on that theme have led to
Ordered Statistics Decoding (OSD) \cite{FL95} and its variants
\cite{gazelle1997reliability,VF04,wu2006soft,wu2007soft,baldi2016use,guo2019some,choi2019fast}.
The principle underlying OSD is to approximate ML decoding
by computing conditional probabilities of the received signal for a
substantially smaller list than the whole codebook, which one hopes contains the ML decoding.
The number of real-valued conditional probabilities that must then
be computed per received signal is determined by the size of the
list. As with ISD, in OSD the linear code's column order is re-arranged
and Gaussian elimination used to systematize it, but rather than
using repeated random permutations the columns are ordered once
in terms of decreasing bit reliability of the received transmission
as determined from the soft information. The most reliable $k$ bits
are hard demodulated and a list of all binary sequences within a
fixed Hamming distance, $t$, of that sequence is created. Each of
those $\sum_{i=0}^t {k \choose i}$ sequences are multiplied by the
revised code generator to create putative code-words in the MRB.
The conditional probabilities of the received sequences for these
code-words, rather than all code-words, is then computed. The most
likely one is identified and converted back to the original basis
as the decoding. OSD's approximation relies on the principle that
if one takes a hard decision on the $k$ most reliable channel
observations, depending on channel conditions, only few errors are
expected within them, with the majority of the errors
introduced by the channel instead being contained within the least
reliable channel outputs, which are essentially ignored for decoding
purposes. The larger the list, the better the approximation to ML
decoding. 

The original hard detection GRAND algorithm \cite{Duffy18,Duffy19,An20} is a true ML block-code decoder 
for hard detection channels subject to noise with or without memory. GRAND's operation requires a 
method to query a string's membership of a codebook.
If the code is unstructured and stored in a dictionary, each query
corresponds to a tree-search with a complexity that is logarithmic
in the code-length. If the code was a Cyclic Redundancy Check (CRC)
code, which is traditionally only used for error detection, checking for
codebook membership requires only a simple polynomial calculation.
If the code is linear in any field, codebook membership can be
determined by a single matrix multiplication and comparison. The
matrix multiplication results in the evaluation of a syndrome, but
GRAND is not a syndrome decoder. No syndrome table is kept and, if
channel conditions change, GRAND naturally adapts its decoding
without recomputing an entire syndrome table. This latter point is
particularly significant in the presence of soft information, which
effectively serves to provide a distinct channel for each communication.


\section{Mathematical Analysis}
\label{sec:math}
As in \cite{Duffy19}, for the analysis of SRGRAND and SRGRANDAB we
exploit the fact that the algorithm is a race between sequential
queries either identifying the noise in the channel, which results
in a correct decoding, or encountering a non-transmitted
element of the codebook, which results in an error. The difference
with SRGRAND is that the decoder is faster and more precise than GRAND 
because it only asks questions of the sub-string that has been potentially
impacted by noise. While the analysis is more involved,
the results obtained are, possibly surprisingly, as clean as in the
hard detection setting. Our mathematical treatment relies on techniques 
from Large Deviation Theory. While we endeavour to provide guiding heuristics, 
to follow the arguments in
detail requires familiarity with that theory \cite{Dembo98,varadhan2008,weiss2019large}.

To analyze the algorithm, we recall notions of guesswork
\cite{Massey94,Arikan96}. Given the receiver is told that
$n$ symbols have been potentially impacted by noise, it creates
a list of noise sequences,
$\GN:\A^n\mapsto\{1,\ldots,|\A|^n\}$, ordered from most likely to least
likely, with ties broken arbitrarily:
$\GN(z^{n,i})\leq \GN(z^{n,j}) \text{ iff } P(\Noise^n=z^{n,i})\geq P(\Noise^n=z^{n,j})$.
For a sequence, $z^n\in\A^n$, its guesswork is the integer $\GN(z^n)$.
For example, if the channel were binary, $\A=\{0,1\}$, and noise
was Bernoulli for some $p<0.5$, then the guesswork order follows Hamming weight. 
For independent and identically distributed
noise on more general alphabets, the family of measures that share the same guesswork order
are described by an exponential family \cite{Beirami19}.

\begin{assumption}[Noise distribution]
\label{ass:N}
When noise occurs, it is independent and identically distributed
as $\UNoise$ where $P(\UNoise=i)=p_{N|S}(i|1)=P(N=i|S=1)$ for $i\in\A$.
\end{assumption}

Under Assumption \ref{ass:N}, if one must guess the entire noise
string of length $n$, Arikan \cite{Arikan96} first established how
the non-negative moments of guesswork, $E(G(N^n)^\alpha)$ for $\alpha>0$, 
scale in $n$ in terms of R\'enyi entropies of order $\alpha$. Building
on those and subsequent results that treated negative moments,
\cite{Pfister04} for $\alpha>-1$ and  for
$\alpha\leq-1$, and more general noise sources, it was established \cite{Christiansen13}
that the logarithm of guesswork satisfies a Large Deviation Principle (LDP) \cite{Dembo98}. The LDP provides estimates on the
distribution of the number of queries required to correctly identify
a noise-string and was used as the basis to analyze one side of the
decoding race in the hard detection setting \cite{Duffy18,Duffy19}. Recall that all logarithms are base $|\A|$.

\begin{proposition}[Guesswork Moments and Large Deviation Principle \cite{Arikan96, Pfister04, Christiansen13}]
Under assumption \ref{ass:N}, if $S^n=1^n$ so that all received
symbols are potentially impacted by noise,  and are distributed as
$\UNoise$, the scaled Cumulant Generating Function (sCGF) of $\{n^{-1}\log
G(N^n)\}$ exists:
\begin{align}
\LambdaN(\alpha)
=\lim_{n\to\infty} \frac 1n \log E(G(N^n)^\alpha | S^n=1^n) 
=\lim_{n\to\infty} \frac 1n \log E(G(\UNoise^n)^\alpha) 
	&= 
	\begin{cases}
	\displaystyle
	\alpha H_{1/(1+\alpha)} & \text{ if } \alpha>-1\\
	-\Hmin & \text{ if } \alpha\leq-1,\\
	\end{cases}
\label{eq:sCGF}
\end{align}
where $H_\alpha$ is the R\'enyi entropy of a single noise element, $\UNoise$, with parameter $\alpha$
\begin{align*}
H_\alpha &= \frac{1}{1-\alpha}\log\left(\sum_{i\in\A} p_{N|S}(i|1)^\alpha\right),
\text{ }
H_1 = \HN = -\sum_{i\in\A} p_{N|S}(i|1) \log p_{N|S}(i|1), 
\text{ and }
\Hmin = -\max_{i\in\A}\log p_{N|S}(i|1).
\end{align*}
Moreover, given $S^n=1^n$, the process $\{n^{-1}\log G(N^n)\}$
satisfies a LDP (e.g. \cite{Dembo98}) with convex rate-function
\begin{align}
\label{eq:rf}
\IN(x) = \sup_{\alpha\in\R}(x\alpha-\LambdaN(\alpha)), \text{where }\IN(0)=\Hmin
\text{ and }\IN(\HN)=0.
\end{align}
\end{proposition}
Heuristically, Eq. \eqref{eq:rf} implies that $P(G(N^n)\approx |\A|^{nx}) = |\A|^{-n\IN(x)}$ 
for large $n$. As $\IN(\HN)=0$, with high probability the number of queries until 
$N^n$ is identified concentrates at $|\A|^{n\HN}$. The probability that $N^n$ is
identified in either fewer or more queries decays exponentially in $n$ with a 
rate governed by the convex function $\IN$ defined in Eq. \eqref{eq:rf}. 
Setting $\alpha=1$ in Eq. \eqref{eq:sCGF}, as Arikan originally
did in his investigation of sequential decoding, establishes that
the expected guesswork grows exponentially in $n$ with rate
$\Hhalf$.
which is greater than the Shannon entropy, $\HN$. That the zero of the
rate-function in Eq. \eqref{eq:rf} occurs at $\HN$ ensures, however, that the majority of the probability
is accumulated by making queries up to and including the Shannon typical
set. The apparent discrepancy in these two facts occurs because the
guesswork distribution has a long tail that dominates its average but has little probability.

In the symbol reliability setting,  
it is not necessary to guess a noise-string of length $n$.
Instead, one must guess a random number of symbols corresponding
to those inside the mask that are potentially noise-impacted. To that end, we have
the following assumption on the size of the mask, which is the number
of potentially noise-impacted symbols per transmission. 
\begin{assumption}[Number of potentially noise-impacted symbols - size of mask]
\label{ass:L}
With $L^n=\sum_{i=1}^nS^n_i$ being the mask size, i.e. the
number of potentially noise-impacted symbols in a block of length $n$,
the proportion of them, 
$\{L^n/n\}$, satisfies a LDP with a strictly
convex rate-function $\IL:\R\mapsto[0,\infty]$ such that $\IL(l)=\infty$
if $l\notin[0,1]$ and $\IL(\muL)=0$, where $\lim_n E(L^n/n) =
\muL>0$. Define the sCGF for $\alpha\in\R$
to be
$
\LambdaL(\alpha) = 
\lim_{n\to\infty} \frac 1n \log E\left(|\A|^{\alpha L^n}\right) = 
\sup_{l\in[0,1]}\left(\alpha x - \IL(l)\right)$,
which exists in the extended reals owing to Varadhan's Lemma (e.g. \cite{Dembo98}[Theorem 4.3.1]).
\end{assumption}
Roughly, Assumption \ref{ass:L}, which is true for a broad class of
processes $\{S^n\}$ including i.i.d., Markov and general mixing, e.g.
\cite{Dembo98}, says  the probability of having $nl$  potentially noise impacted  symbols   decays exponentially
in $n$ with a rate, $\IL(l)$, that is positive unless $l$ is the mean
$\muL$, i.e. $P(L^n\approx nl) \approx |\A|^{-n\IL(l)}$. 

With some abuse of notation for Shannon entropy, under Assumptions
\ref{ass:N} and \ref{ass:L}, recalling that we define all logarithms as base $| \A |$, the symbol reliability decoding channel's capacity, $\CSOFT$
is upper bounded by
\begin{align}
\CSOFT &\leq \limsup_{n\to\infty} \frac 1n \sup I(X^n;(Y^n,S^n)) 
	\leq 1- \limsup_{n\to\infty} \frac1n H\left(\UNoise^{L^n}\right)  
	= 1- \limsup_{n\to\infty} \frac{E(L^n)}{n} H(\UNoise)\nonumber\\
	&= 1-\muL h(p_{N|S}(\cdot|1)), \label{eq:CSOFT}
\end{align}
where $h\left(p_{N|S}(\cdot|1)\right)=-\sum_{i\in\A}p_{N|S}(i|1)\logA p_{N|S}(i|1)$
is the Shannon entropy of $p_{N|S}(\cdot|1)$,
we have upper-bounded the entropy of the input by $1$,
and used the fact that the channel is invertible (i.e.
Eq. \eqref{eq:channel}). Through  constructing
SRGRAND and SRGRANDAB, we will show $\CSOFT$ is attainable.

Under Assumptions \ref{ass:N} and \ref{ass:L}, in a distinct context
and for a distinct purpose, it was established in \cite{Christiansen13b}
that with a random number of characters to be guessed one
has the following LDP.

\begin{proposition}[LDP for guessing subordinated noise \cite{Christiansen13b}]
\label{theorem:LDPN}
Under assumptions \ref{ass:N} and \ref{ass:L},
the joint subordinated guesswork and length process $\{(1/n \log G(\UNoise^{L^n}), L^n/n)\}$ 
satisfies a LDP with the jointly convex rate-function
\begin{align}
\label{eq:INLjoint}
\INLjoint(g,l) = l\IN\left(\frac{g}{l}\right)+\IL(l),
\end{align}
where $\IN$ is the guesswork rate-function defined in Eq. \eqref{eq:rf} and
$\IL$ is the length rate-function defined in Assumption \ref{ass:L}. Note that
$\INLjoint(H,\muL)=0$, where $H$ is Shannon entropy of a noise-impacted symbol and
$\muL$ is the average number of potentially noise-impacted symbols.

The subordinated guesswork process $\{1/n \log
G(\UNoise^{L^n})\}$ 
satisfies a LDP with the convex rate
function
\begin{align}
\label{eq:INL}
\INL(g) = \inf_{l\in[0,1]}
	\left(l\IN\left(\frac{g}{l}\right)+\IL(l)\right),\text{ where }\INL(\muL\HN)=\INLjoint(H,\muL)=0.
\end{align}

The sCGF for $\{1/n \log G(\UNoise^{L^n})\}$, the Legendre-Fenchel
transform of $\INL$, is given by the composition of the
sCGF for the length with the sCGF for the guesswork of non-subordinated noise
\begin{align}
\label{eq:LambdaNL}
\LambdaNL(\alpha) 
&
=\lim_{n\to\infty} \frac 1n \log E\left(G\left(\UNoise^{L^n}\right)^\alpha\right)
	= \LambdaL(\LambdaN(\alpha))
	= \sup_g\left(g\alpha-\INL(g)\right) \text{ for } \alpha\in\R.
\end{align}
In particular, the average number of queries required to identify subordinated noise is
given by
\begin{align}
\label{eq:lambdaNL1}
\LambdaNL(1)  = \lim_{n\to\infty} \frac 1n \log E\left(G\left(\UNoise^{L^n}\right)\right)
 = \LambdaL(\Hhalf).
\end{align}
\end{proposition}
Roughly speaking, the joint LDP indicates that, for large $n$,
$
P\left(\left(\frac1n \log G\left(\UNoise^{L^n}\right), \frac{L^n}{n}\right)\approx(g,l)\right) 
\approx |\A|^{-n \INLjoint(g,l)},$
and $\INLjoint(g,l)$ in Eq. \eqref{eq:INLjoint} can be interpreted
as follows: if the number of potentially noise-impacted symbols is
$L^n\approx nl$, which is exponentially unlikely with rate $\IL(l)$,
then having  the logarithm of the subordinated guesswork  be $\log
G\left(\UNoise^{L^n}\right)\approx ng$ has essentially the same likelihood
as $\log G\left(\UNoise^{ln}\right)\approx ng$, which has rate $l\IN(g/l)$
as a total deviation of $g$ must be accrued over a smaller proportion
of potentially noise-impacted symbols. The unconditioned LDP follows
from the large deviations mantra that rare events occur in the most
likely way, so that the rate-function $\INL$ is determined from the
proportion of potentially noise-impacted symbols that gives the
smallest decay rate for the probability.

Results on the subordinated guesswork process $\{1/n \log
G(\UNoise^{L^n})\}$ governed by the rate-function in Eq.
\eqref{eq:INL} are sufficient to enable us to prove a Channel
Coding Theorem for the symbol reliability channel. Finer-grained results
on error exponents that depend on the proportion of symbols that
were noise-impacted, however, follow from the LDP for the joint subordinated guesswork and length process
$\{(1/n \log G(\UNoise^{L^n}), L^n/n)\}$ governed by the rate-function
given in Eq. \eqref{eq:INLjoint}.

We note that $\LambdaL$ is a convex function whose derivative at the
origin is $\muL$, the mean number of potentially noise-impacted
symbols, so that $\LambdaL(\Hhalf)\geq \muL \Hhalf$. Hence, from
Eq. \eqref{eq:lambdaNL1}, the average number of queries until
the true channel-noise is identified grows exponentially in $n$ at a
potentially larger rate than the guesswork required for the average
proportion of potentially noise-impacted symbols. Despite that, the
zero of the rate-function in Eq. \eqref{eq:INL} occurs at $\muL
H$, so that the majority of the likelihood of identifying the true
subordinated noise occurs by the Shannon entropy of the typical set
of average number of potentially noise-impacted symbols. Thus, while
stochastic fluctuations in the number of potentially noise-impacted
symbols has relevance to complexity and error exponents, that
variability has no impact on capacity. In a manner akin to GRANDAB, without
loss of capacity, 
complexity can be ameliorated
by abandoning guessing after a suitable number of queries.

To mathematically characterize the number of queries made until
a non-transmitted code-word is identified, which is
the second part of the guesswork decoding race, we assume that the
codebook is created uniformly at random.
For uniformly distributed codebooks, the location of each 
element in the guessing order of a received transmission is itself
uniform in $\{1,\ldots,|\A|^n\}$. The distribution of
the number of guesses until any non-transmitted element of the
codebook is hit upon is thus distributed as the minimum of $M_n$ 
uniform random variables. We can, therefore, 
use the
following result from \cite{Duffy18, Duffy19}, again recalling that our logarithm is base $|\A|$.

\begin{proposition}[LDP for Guessing a Non-transmitted Code-word \cite{Duffy18, Duffy19}]
\label{theorem:LDPU}
Assume that $M_n=\lfloor |\A|^{nR}\rfloor$ for some $R>0$,
and that $U^{n,1},\ldots,U^{n,M_n}$ are independent random
variables, each uniformly distributed in $\{1,\ldots,|\A|^n\}$. 
Defining $U^n =\min_i U^{n,i}$, $\{1/n \log U^{n}\}$ then $
\lim_{n\to\infty} \frac 1n \log \E(U^n) = 1-R.
$ and  $U^n$ satisfies
a LDP with the lower semi-continuous rate-function
\begin{align}
    \IU(u) = \begin{cases}
    1-R-u & \text{ if } u\in[0,1-R]\\
    +\infty & \text{ otherwise.}
    \end{cases}
\label{eq:IU}
\end{align}

\end{proposition}

\begin{figure*}
\begin{center}
\vspace{-0.2in}
\includegraphics[width=0.5\textwidth]{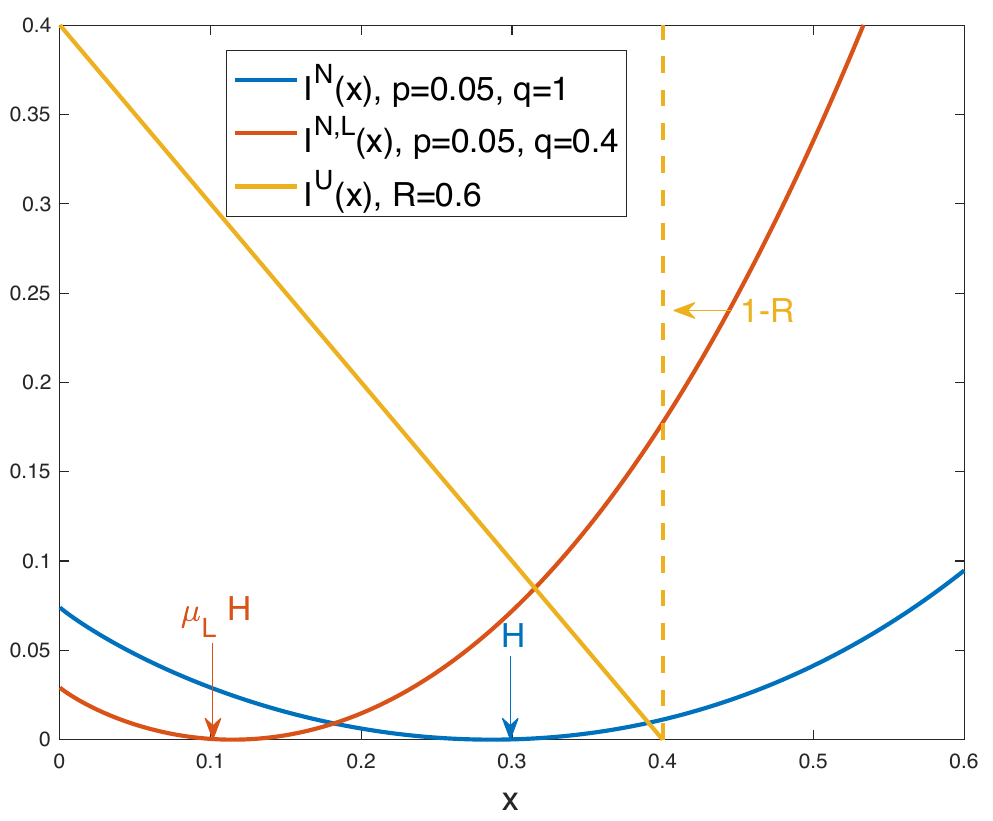}
\vspace{-0.2in}
\end{center}
\caption{Probabilistic guesswork decoding race in a SR-BSC. With $p=0.05$ and
codebook rate $R=0.6$, the large deviations rate
function for: incorrectly identifying a non-transmitted
element of the codebook, $\IU(x)$; guessing the true noise if 
$q=1$ and all bits are potentially noise-impacted, $\IN(x)$; 
with $q=0.4$, guessing the true noise if a random set of locations are potentially noise-impacted, $\INL(x)$.
With $x$ being the value
on the x-axis, when $2^{nx}$ noise guesses are made the 
likelihood of success for each of these three racing elements is approximately
$2^{-n\inf_{y<x}I(y)}$ for the relevant rate function, $I(y)$.
}
\vspace{-0.2in}
\label{fig:1}
\end{figure*}

A graphical representation of the rate-functions that determine the
asymptotic likelihoods of outcomes of this guessing race can be
found in Fig.~\ref{fig:1}. When all symbols are subject to noise,
as in \cite{Duffy18,Duffy19}, the channel is within capacity so
long as the zero of the rate-function for guessing noise, which
occurs at the Shannon entropy rate of the noise $\HN$, is smaller
than the zero of the rate-function for identifying a non-transmitted
code-word, which occurs at $1-R$, where $R$ is the normalized
codebook rate.  As in all likelihood the correct decoding is
identified after fewer queries than an incorrect element of the
codebook would be identified, the algorithm experiences concentration
onto a correct decoding, which leads to the proof of the classical
hard detection Channel Coding Theorem, $R<1-\HN$, in
\cite{Duffy18,Duffy19}. In the present paper, the zero of the rate
function for the subordinated noise-guessing occurs at $\muL\HN$,
the average number of potentially noise-impacted symbols times the
Shannon entropy of the noise. So long as $\muL\HN$ is smaller than
$1-R$, noise-guessing concentrates on identifying correct decodings
before erroneous ones, leading to the Symbol Reliability
Channel Coding Theorem, proved below, where any $R<1-\muL\HN$ is
achievable.

The proportion of potentially noise-impacted symbols is available
to the receiver and so it is reasonable to consider error exponents
subject to its knowledge.  We characterize these error exponents in terms of 
$R$ and the rate-function $\INLjoint$ given in Eq.
\eqref{eq:INLjoint}.In particular, define
\begin{align}
\label{eq:errorRL}
\errorL(R,l) = -\lim_{\delta\downarrow0}\lim_{n\to\infty}
	\frac1n \log P\left( \frac{L^n}{n}\in(l-\delta,l+\delta),
	U^n\leq\GN\left(\UNoise^{L^n}\right)\right)
\end{align}
to be the probability exponent that the proportion of potentially
noise-impacted symbols, representing the size of the mask, is $l$, and that there is an error, as the
number of queries required to identify a non-transmitted code-word
is smaller than the number of queries required to identify the true
noise.

\begin{theorem}[Symbol Reliability Channel Coding Theorem]
\label{prop:channel_coding_theorem}
Assuming 
$R<1-\muL\HN$, 
under Assumptions \ref{ass:N} and \ref{ass:L}, and those of
Proposition \ref{theorem:LDPU},
we have that the
probability that the conditional ML decoding  of SRGRAND is incorrect
decays exponentially in  $n$,
\begin{align}
\label{eq:noerror}
    \lim_{n\to\infty} \frac 1n \log P\left(U^n\leq\GN\left(\UNoise^{L^n}\right)\right)=
-\inf_{u\in[\muL\HN,1-R]}\{\IU(u)+\INL(u)\}<0.
\end{align}
If 
$\gstar$ exists such that \vspace{-0.3in}
\begin{align}
\label{eq:gstar}
\frac{d}{dg}\IN(g)|_{g=\gstar}=1,
\end{align}
which is analogous to one minus Gallager's critical rate,
then the joint error exponent of \eqref{eq:errorRL},
subject to a given proportion of potentially noise-impacted symbols
satisfies
\begin{align}
\label{eq:errorL}
    \errorL(R,l) &
	= \begin{cases}
	\IL(l)+1-R-l\Hhalf & \text{ if } R\in(0,1-l\gstar]\\
	\displaystyle \IL(l)+l\IN\left(\frac{1-R}{l}\right) & \text{ if } R\in[1-l\gstar,1-l\HN]\\
	\IL(l) & \text{ if } R\in(1-l\HN,1].
	\end{cases}
\end{align}
The unconditioned SRGRAND error rate
is
\begin{align}
\vspace{-0.3in}
\label{eq:error}
\error(R)=\inf_{l\in[0,1]}\errorL(R,l)
     &= -\lim_{n\to\infty} \frac 1n \log P\left(U^n\leq\GN\left(\UNoise^{L^n}\right)\right) 
	= \begin{cases}
	1-R-\LambdaL(\Hhalf) & \text{ if } R\in(0,1-\muL\gstar)\\
	\INL(1-R) & \text{ if } R\in[1-\muL\gstar,1-\muL\HN)\\
	0 & \text{ if } R\in(1-\muL\HN,1].
	\end{cases}
\end{align}

With $\delta>0$, abandoning guessing if $|\A|^{n (\muL\HN +\delta)}$ queries
have been made without identifying an element of the codebook, the
SRGRANDAB error rate is also negative,
\begin{align}
&\lim_{n\to\infty} \frac 1n \log P\left(\left\{U^n\leq\GN(\UNoise^{L^n})\right\}
	\cup\left\{\GN(\UNoise^{L^n})\geq |\A|^{n(\muL\HN+\delta)}\right\}\right)\nonumber\\
&=-\min\left(\inf_{u\in[\muL\HN,1-R]}\{\IU(u)+\INL(u)\},\INL(\muL\HN+\delta)\right)<0.
\label{eq:errorAB}
\end{align}
If, in addition, $\gstar$ defined in Eq. \eqref{eq:gstar} exists 
then the expression simplifies to 
 $  \errorAML(R) = \min\left(\error(R), \INL(\HN+\delta)\right)<0$
where $\error(R)$ is the conditional ML decoding error rate in Eq. \eqref{eq:error}.

\end{theorem}
\begin{proof}
As $\{U^n\}$ is independent of $\{(G(\UNoise^{L^n}),L^n)\}$, we have that
$\left\{\left(n^{-1} \log U^n, n^{-1} \log G(\UNoise^{L^n}), L^n/n\right)\right\}$
satisfies a LDP with rate-function $\IU(u)+\INLjoint(g,l)$. 
Noting the equivalence of the following two events,
\begin{align*}
\left\{U^n\leq\GN\left(\UNoise^{L^n}\right)\right\} =
\left\{\frac1n \log\left(U^n/\GN\left(\UNoise^{L^n}\right)\right)\leq0\right\}.
\end{align*}
By the contraction principle (e.g. \cite{Dembo98}[Theorem 4.2.1])
with the continuous function $f(u,g,l)=(u-g,l)$, the process 
$
\left\{
\left(\frac1n
\log\left(U^n/\GN\left(\UNoise^{L^n}\right)\right),
\frac{L^n}{n}\right)\right\}
$
satisfies a LDP with rate-function 
$\inf_{u\in[0,1-R]}\left\{\IU(u)+\INLjoint\left(u-x,l\right)\right\}$.

Consider $\errorL(R,l)$ defined in Eq. \eqref{eq:errorL},
where the limits exist as the rate-functions are convex and so
continuous on the interior of where they are finite,
\begin{align*}
\errorL(R,l) &= -\lim_{\delta\downarrow0}\lim_{n\to\infty}
	\frac1n \log P\left(\frac1n \log\left(U^n/\GN\left(\UNoise^{L^n}\right)\right)\leq0,
		\frac{L^n}{n}\in(l-\delta,l+\delta)\right)\\
	&= \inf_{x\leq0}\inf_{u\in[0,1-R]}\left\{\IU(u)+\INLjoint\left(u-x,l\right)\right\}
	= \inf_{u\in[0,1-R]}\left\{\IU(u)+\inf_{g\geq u} l\IN\left(\frac{g}{l}\right)\right\}+\IL(l).
\end{align*}
This final expression essentially encapsulates that the error
exponent is the exponent for the likelihood that the proportion of
potentially noise-impacted symbols, or size of the mask, is $l$, plus the smallest exponent
(corresponding to the most likely event) for the minimum of the
scaled uniforms being at $u$, while the scaled sub-ordinated guesswork
occurs at any value $g$ at least as large as $u$.

For $u\in[0,1-R]$, $\IU(u)=1-R-u$ is linearly decreasing, while
$l\IN(g/l)$ is convex in $g$ with minimum, zero, at $g=l\HN$. Thus
if $R\geq 1-l\HN$, setting $u=1-R$ and $g=l\HN$, $\errorL(R,l) =
\IL(l)$. 
If, alternatively, $R<1-l\HN$, then as both $\IU(u)$ and $l\IN(g/l)$, as a function
of $g$, are strictly decreasing on $[0,l\HN]$,
\begin{align*}
\inf_{u\in[0,1-R]}\left\{\IU(u)+\inf_{g\geq u} l\IN\left(\frac{g}{l}\right)\right\}
= \inf_{u\in[l\HN,1-R]}\left\{1-R-u+l\IN\left(\frac{u}{l}\right)\right\},
\end{align*}
which is strictly positive as $\IU$ is strictly decreasing to $0$
on $[l\HN,1-R]$ while $l\IN(u/l)$ is strictly increasing in $u$ on the
same range. Assuming $\gstar$ defined
in Eq. \eqref{eq:gstar}, exists, as $\IU$ is decreasing at
rate $1$ and
$
\frac{d}{dg}l\IN\left(\frac{g}{l}\right)|_{g=l\gstar}=1,
$
then if $l\gstar\leq1-R$, i.e. if $R\leq 1-l\gstar$, 
\begin{align*}
\inf_{u\in[l\HN,1-R]}\left\{1-R-u+l\IN\left(\frac{u}{l}\right)\right\}
= 1-R-l\gstar+l\IN\left(\frac{l\gstar}{l}\right)
= 1-R-l\gstar+l\IN\left(\gstar\right)
= 1-R-l\Hhalf,
\end{align*}
as $\IN(\gstar)=\gstar-\Hhalf$. If, instead, $l\gstar\geq1-R$, then the infimum
occurs at $u=1-R$ and
\begin{align*}
\inf_{u\in[l\HN,1-R]}\left\{1-R-u+l\IN\left(\frac{u}{l}\right)\right\}
=
l\IN\left(\frac{1-R}{l}\right) \text{ if } R\in[1-l\gstar,1-l\HN],
\end{align*}
and the expression in \eqref{eq:errorL} follows. 
The unconditional error exponent, $\error(R)$ in Eq.
\eqref{eq:error}, is obtained from that in \eqref{eq:errorL} by the contraction
principle, projecting out $L^n/n$, giving $\error(R)=\inf_{l\in[0,1]}\errorL(R,l)$. 
If $R\geq 1-\muL\HN$, then $\error(R)=\errorL(R,\mu^L)=0$.
If $R\in[1-l\gstar,1-l\HN]$, then 
$
\error(R) = \inf_l\left\{\IL(l)+l\IN\left(\frac{1-R}{l}\right)\right\} = \INL(1-R).
$
Finally, if $R\in(0,1-l]$, then
$
\error(R) = \inf_l \left\{\IL(l)+1-R-l\Hhalf\right\}
= (1-R)-\inf_l\left\{l\Hhalf-\IL(l)\right\} = 
	1-R-\LambdaL(\Hhalf),
$
inverting the Legendre-Fenchel transform in the last step.

To determine the error exponent of SRGRANDAB, by the Principle of
the Largest Term \cite[Lemma 1.2.15]{Dembo98} it suffices to
consider only the smallest of the two exponential rates in Eq.
\eqref{eq:errorAB}. The first term is the error rate for GRAND. The
 second term is the exponent of the probability of
error due to abandonment of guessing. Note that
$
P\left(\GN(\UNoise^{L^n})\geq |\A|^{n(\muL\HN+\delta)}\right)
=P\left(\frac1n\log \GN(\UNoise^{L^n})\geq \muL\HN+\delta\right)
$
and the result follows from the LDP as $\INL(x)$ is convex and
increasing for $x>\muL\HN$.
\end{proof}

Consider the error exponent for the conditional ML decoding 
via SRGRAND, $\errorL(R,l)$ in \eqref{eq:errorL}. The exponent for the likelihood that the proportion
of potentially noise-impacted symbols, $L^n/n$,  which is approximately
$l$, is $\IL(l)$.  The error-exponent is  then as in a channel where only a proportion $l$ of transmitted
symbols are in the mask of symbols  subject to noise \cite{Duffy19}. The unconditional
equivalent, $\error(R)$ in Eq. \eqref{eq:error} identifies the
most likely proportion of noise-impacted symbols that may give rise to an error for a given codebook rate. For SRGRANDAB,
an error occurs either if the identified conditional ML decoding is in error
or if abandonment occurs. The more likely of these two events dominates in the limit.

Combining Propositions \ref{theorem:LDPN} and \ref{theorem:LDPU}
in a distinct way enables us to determine the asymptotic complexity of the
SRGRAND and SRGRANDAB in terms of the number of queries until a
decoding, correct or incorrect, is identified: 
$
\GML^n := \min\left(\GN\left(\Noise^{L^n}\right),U^n\right). $
That is, the algorithm terminates when the channel noise
or a non-transmitted element of
the codebook are identified, whichever occurs first.
On the scale of large deviations, if the codebook is within capacity,
$R<1-\muL\HN$, then it becomes apparent that the sole impact
of the codebook is to curtail excessive guessing when unusual noise
occurs.
The number of guesses SRGRANDAB makes until terminating is 
$
\GAML^n := \min\left(\GN\left(\Noise^{L^n}\right),U^n,|\A|^{n(\muL\HN+\delta)}\right). 
$
The final term corresponds to the abandonment threshold, curtailing
guessing shortly after 
the Shannon typical set for an average
number of potentially noise impacted symbols.

\begin{theorem}[Complexity of SRGRAND and SRGRANDAB]
\label{prop:IMLE}
If $R<1-\muL\HN$,
under Assumptions \ref{ass:N} and \ref{ass:L}, and those
of Proposition \ref{theorem:LDPU}, the scaled complexity
of SRGRAND, $\{1/n \log \GML^n\}$,
satisfies the LDP with a convex  rate-function
\vspace{-0.2in}
\begin{align}
    \IMLE(d) = 
	\begin{cases}
    \INL(d) & \text{ if } d\in[0,1-R]\\
    +\infty & \text{ if } d>1-R
    \end{cases}
\label{eq:IMLE}
\end{align}
and the expected number of guesses for SRGRAND to find a conditional ML decoding satisfies
$
	\lim_{n\to\infty}\frac 1n \logA E(\GML^n) = \min\left(\LambdaL(\Hhalf),1-R\right).$
With $\delta>0$, the complexity of SRGRANDAB,
$\{1/n \log \GAML^n\}$, satisfies a LDP with a convex rate function
\begin{align}
    \IMLEAB(d) = 
	\begin{cases}
  \INL(d) & \text{ if } d\in[0,\min(1-R, \muL\HN)]\\
    +\infty & \text{ if } d>\min(1-R, \muL\HN)
    \end{cases}
\label{eq:IMLEAB}
\end{align}
and the expected number of guesses until SRGRANDAB terminates, $\{\GAML^n\}$, satisfies
$
\lim_{n\to\infty} \frac1n \logA E(\GAML^n) = \min\left(\LambdaL(\Hhalf),1-R,\mu\HN+\delta\right).$
\end{theorem}
\begin{proof}
Consider the process $\{n^{-1}\log \GML^n\}$, following \cite{Duffy19}[Proposition 2],
as $f(g,u)=\min(g,u)$ is a continuous function, by the contraction principle it 
satisfies a LDP with rate-function
$\IMLE(d) = \inf\{\INL(g)+\IU(u):\min(g,u)=d\}$. 
If $d>1-R$, $\IMLE(d)=\infty$ as $\IU(d)=\infty$ for $d>1-R$. Alternatively,
if $d\leq 1-R$,
\begin{align*}
\IMLE(d) = \min\left(\INL(d)+\inf_{x\geq d}\IU(x), \inf_{x\geq d}\INL(x)+\IU(d)\right)
= \min\left(\INL(d), \inf_{x\geq d}\INL(x)+\IU(d)\right)
\end{align*}
as $\IU(x)$ is decreasing for $x\in[0,1-R]$. If $R<1-\muL\HN$, then
note the geometric consideration
\begin{align*}
\INL(0)
=\inf_{l}\left\{\IL(l)+l\IN(0)\right\} 
= \inf_l\left\{\IL(l)+l\Hmin\right\} 
\leq\muL \Hmin,
\end{align*}
where in the last inequality we have set $l=\muL$.  As min-entropy
is less than Shannon entropy $\muL\Hmin\leq \muL\HN<1-R$ and as
$\INL$ is convex, $\INL(d)\leq \IU(d)$ for all $d\in[0,\HN]$ while
$\INL(d)$ is increasing on $[\HN,1-R]$ and so $\IMLE(d)=\INL(d)$
for $d\in[0,1-R]$.

To obtain the scaling result for $E(\GML^n)$ we invert the
transformation from the rate function $\IMLE$ to its Legendre-Fenchel
transform, the sCGF of the process
$\{n^{-1}\log \GML^n\}$ via Varadhan's Theorem \cite{Dembo98}[Theorem
4.3.1]. In particular, note that, regardless of whether $\IMLE$ is
convex or not,
\begin{align*}
        \lim_{n\to\infty}\frac 1n \log E(\GML^n) &=
        \lim_{n\to\infty}\frac 1n \log E\left(|\A|^{\log\GML^n}\right) 
        = \sup_{d\in\R} \{d-\IMLE(d)\} = \min\left(\LambdaL(\Hhalf),1-R\right).
\end{align*}

The final component of the minimum
satisfies an LDP with a rate function
$0  \text{ if } d=\mu+\delta$ and 
$+\infty  \text{ if } d\neq\mu+\delta$. 
and, again, as minimum is continuous by the contraction principle
the LDP with a rate-function given in Eq. \eqref{eq:IMLEAB}
and the scaling of $E(\GAML^n)$ follows from similar considerations.
\end{proof}

Theorem \ref{prop:IMLE} effectively says that in SRGRAND the algorithm
terminates with a correct decoding so long as the number of queries
made before identifying an element of the codebook is less than
$|\A|^{n(1-R-\epsilon)}$ for some $\epsilon>0$.  If more queries than that are
made, the conditional ML decoding will be erroneous.  SRGRAND queries until it  identifies the true noise or until
an erroneous identification, whichever comes first. In this realization
of SRGRANDAB, querying is abandoned for noise sequences beyond the
typical set of the average number of potentially noise impacted
symbols, curtailing complexity. %

\section{Mathematical Example: Symbol Reliability Binary Symmetric Channel (SR-BSC)} 
\label{sec:SR-BSC}

We consider a setting where it is possible to mathematically
compare
channels with and without knowledge of the symbol reliability
information vector $S^n$, the Symbol Reliability Binary Symmetric
Channel (SR-BSC). For the SR-BSC, we assume that each transmitted
symbol is potentially impacted independently by noise with probability
$p_S(1)=q\in[0,1]$. Code-book and noise symbols take values in a
binary alphabet $\A=\{0,1\}$, $\oplus$ is addition in $\F2$, and
thus $0$ represents the no-noise character. Given a symbol has been
potentially noise-impacted, 
the conditional probability
that the corresponding bit has been flipped is $p_{N|S}(1|1)=p\in[0,1]$,
$p_{N|S}(0|1)=1-p$ and $p_{N|S}(0|0)=1$. The overall bit-flip probability of the SR-BSC is thus $pq$. We consider capacity and
 error exponents, which are properties of ML decoding no
matter whether it is identified by the noise-guessing methodology
or by brute force, as well as
complexity, which is a feature of the noise-guessing approach.  From Eq.
\eqref{eq:CSOFT}, the capacity of the symbol reliability channel is
$\CSOFT(q,p) = 1-q \h2(p),$
where $\h2(p) = -(1-p)\logT(1-p)-p\logT(p)$ is the binary Shannon entropy.
The corresponding hard detection channel
is a BSC with probability $P(N=1)=P(N=1|S=1)P(S=1)=pq$
and so the hard detection channel capacity is
$\CHARD(q,p) = 1- \h2(pq). $
As $\h2$ is concave, $\CSOFT(q,p)\geq \CHARD(q,p)$ for all $q$
and $p$, and so the capacity of the channel with symbol reliability
information is necessarily higher. Depending on the parametrization,
the symbol reliability channel's capacity can be several orders
of magnitude larger than the hard detection capacity.

As the symbol reliability information is constructed of i.i.d.
elements, the rate function governing the LDP for the proportion
of noise impacted symbols, $\{L^n/n\}$ in Assumption \ref{ass:L},
is the Kullback-Leibler divergence, 
$
\IL(l) =
 -(1-l) \logT\left(\frac{1-l}{1-q}\right) 
 -l \logT\left(\frac{l}{q}\right)$,
which has the corresponding sCGF
\begin{align}
\label{eq:BernLambdaL}
\LambdaL(\alpha) = \logT\left(1-q+q 2^{\alpha}\right).
\end{align}
The rate function for LDP of the rescaled guesswork $\{1/n\logT
G(N_1^n)\}$ in Eq. \eqref{eq:rf} is the Legendre-Fenchel
transform, $\IN(g) = \sup_\alpha\left(\alpha g - \LambdaN(\alpha)\right)$,
of 
\begin{align}
\label{eq:BernLambdaN}
\LambdaN(\alpha) = 
\begin{cases}
-\logT\max(p,1-p) & \text{ if } \alpha\leq -1 \\
-p\logT(p)-(1-p)\logT(1-p) & \text{ if } \alpha=1\\
(1+\alpha)\logT\left(p^{1/(1+\alpha)}+(1-p)^{1/(1+\alpha)}\right) & \text{ if } \alpha \in(-1,1)\cup(1,\infty).\\
\end{cases}
\end{align}
From Eq. \eqref{eq:LambdaNL}, the sCGF for the subordinated guesswork of true noise is $\LambdaNL(\alpha) = \LambdaL(\LambdaN(\alpha))$,
where $\LambdaL$ and $\LambdaN$ are given by Eq.
\eqref{eq:BernLambdaL} and Eq. \eqref{eq:BernLambdaN}, respectively.
The exponent of the average complexity required to identify 
the true noise in the symbol reliability channel is given by
$\lim_{n\to\infty} n^{-1} \logT E\left(G\left(N_1^{L^n}\right)\right) 
= \LambdaNL(1) 
= \LambdaL(\LambdaN(1)) 
= \logT\left( 1-q + q 2^{2 \logT(p^{1/2} +(1-p)^{1/2})}\right)$,
while for the hard detection channel it is
$
\lim_{n\to\infty} n^{-1} \logT E(G(N^n)) = 2 \logT\left((pq)^{1/2} +(1-pq)^{1/2}\right).
$

\begin{figure*}
\begin{center}
\includegraphics[width=0.8\textwidth]{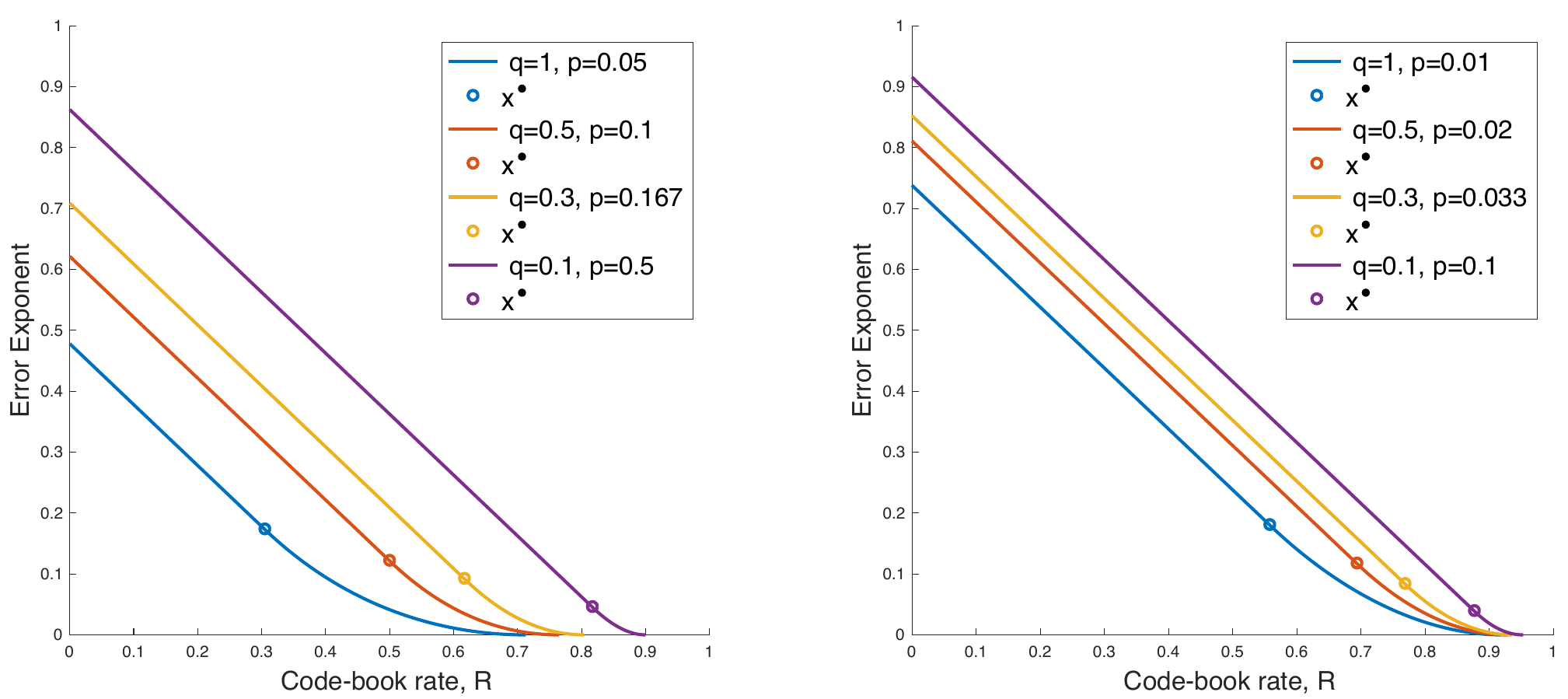}
\vspace{-0.2in}
\end{center}
\caption{Block error exponent comparison between BSCs with and
without symbol reliability information. 
$pq$, the overall bit-flip probability, is constant. Error exponents are
plotted as a function of codebook rate $R$. In the left hand side
plot $pq=0.05$, and on the right-hand side $pq=0.01$. Circles
indicate Gallager's critical rate. The lowest line has $q=1$ and
is the error exponent of the hard detection channel. Higher lines
correspond to different $(q,p)$ combinations and have larger error
exponents, meaning decoding errors are less likely.}
\vspace{-0.2in}
\label{fig:4}
\end{figure*}

Armed with the sCGFs for the proportion of potentially noise impacted
bits and for the rescaled logarithm of the guesswork of potentially
noise impacted bits, the asymptotic error exponent given in
\eqref{eq:error} is readily computable numerically. Recall that,
as a function of the codebook rate $R$, this is the exponent in
the decay rate in the likelihood than a conditional ML decoding is in error
as the block length increases.


\begin{figure*}
\begin{center}
\includegraphics[width=0.4\textwidth]{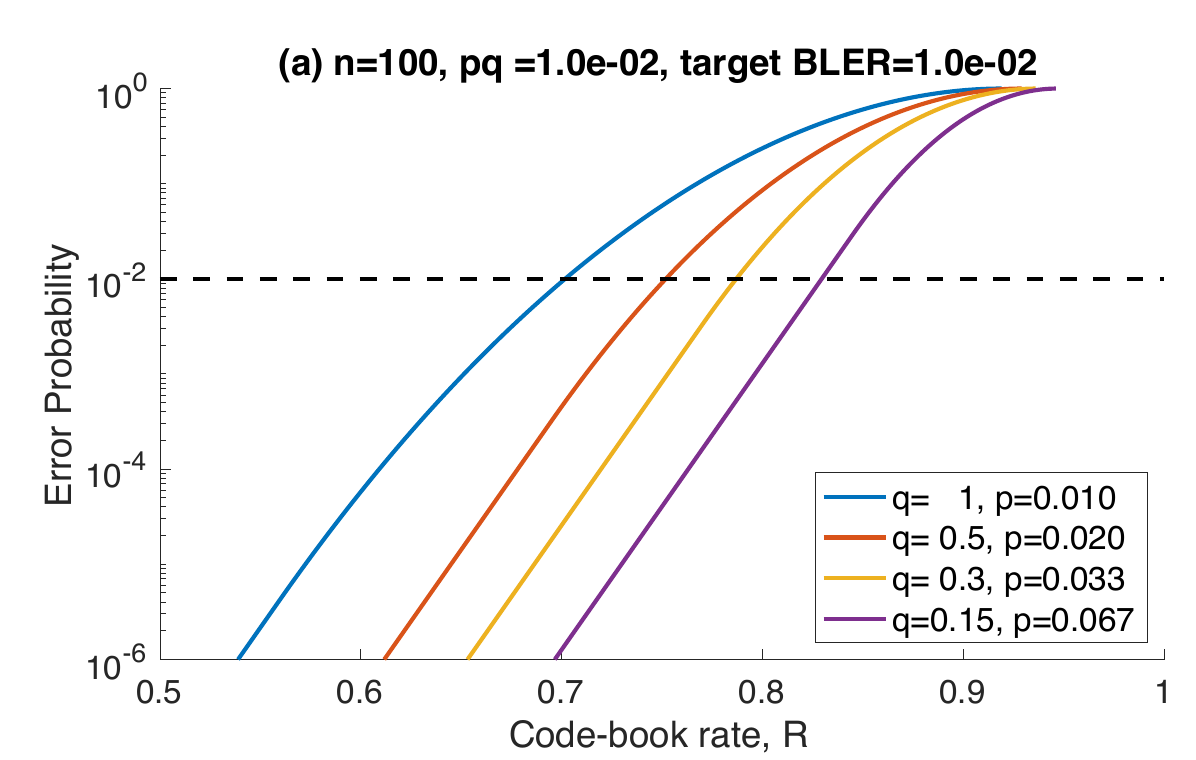}
\includegraphics[width=0.4\textwidth]{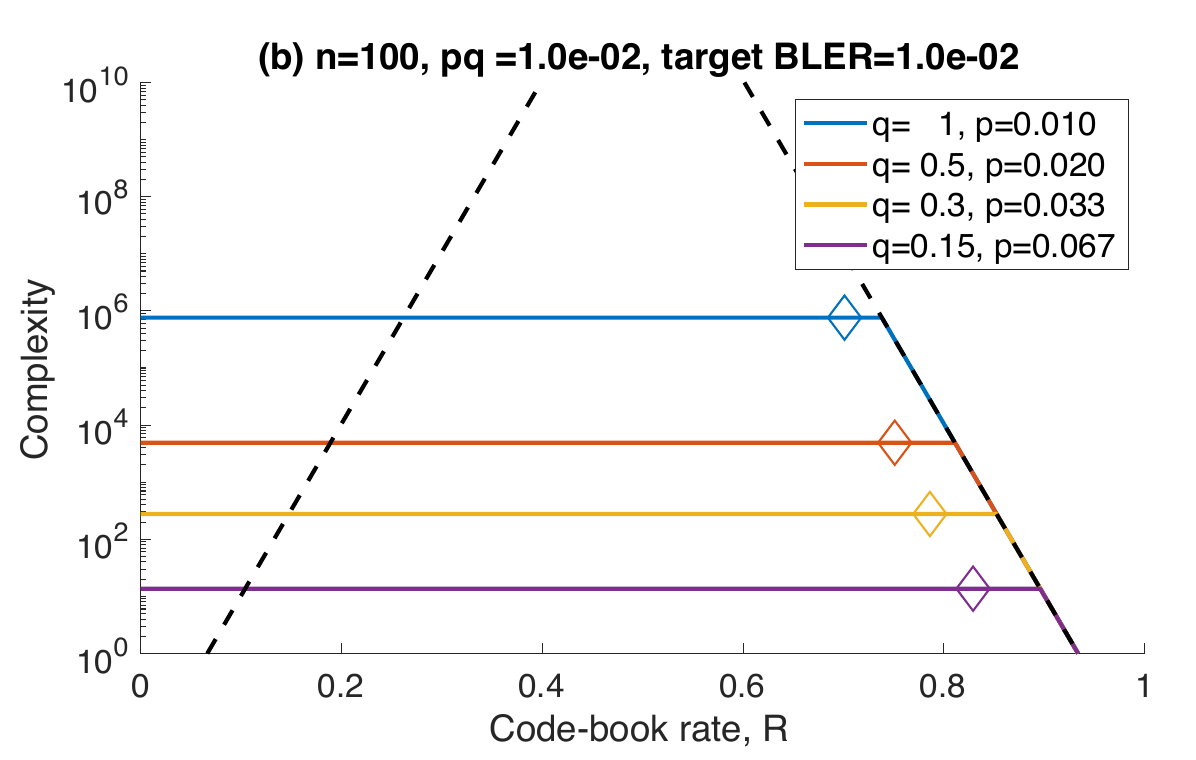}\\
\includegraphics[width=0.4\textwidth]{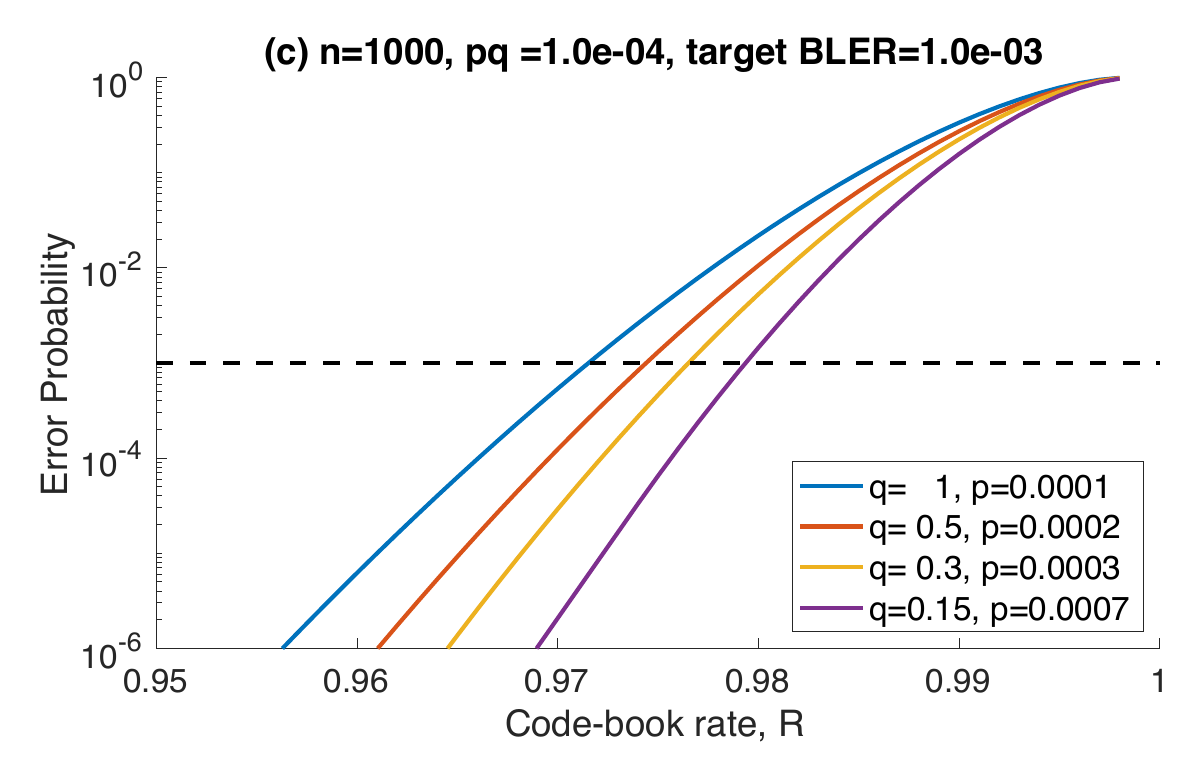}
\includegraphics[width=0.4\textwidth]{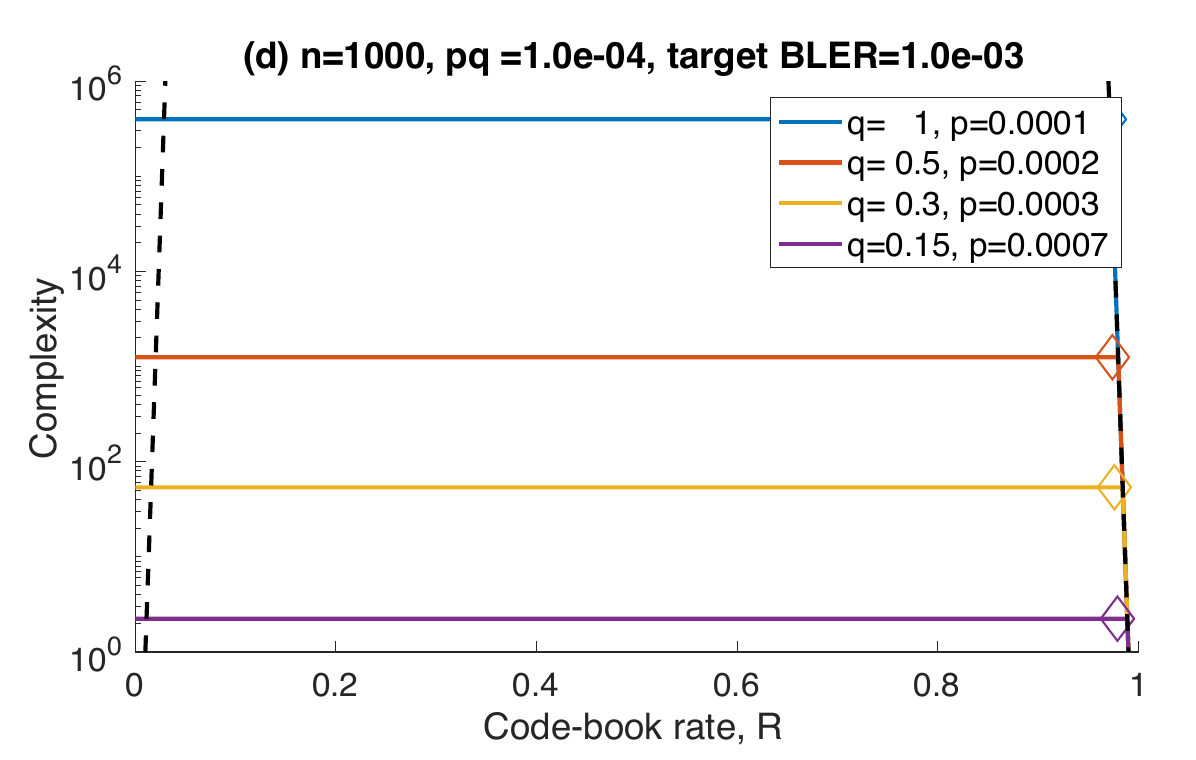}
\vspace{-0.2in}
\end{center}
\caption{Approximate block error probability and complexity for
 BSCs with and without symbol reliability information. The BSC without symbol reliability corresponds to the SR-BSC with $q=1$, and
 the overall bit-flip probability, $pq$, is constant.
(a-b) Show
results for $n=100$, $pq=10^{-2}$ and a target block error of
$10^{-2}$. In (a), the horizontal dashed line is the
target block error and approximate block error probabilities are
shown as a function of codebook rate, $R$, for a selection of
$(q,p)$ pairs. (b) shows the approximate complexity, in terms of average
number of guesses per-bit to
identification of a codebook element, which decreases with $q$,  even though
$p$ is increasing. The dashed black line gives 
complexity for the 
approach of \cite{BCJR74}, 
Diamonds indicate the rate above which
the target block error rate would be exceeded, 
while the inflection point occurs at cut-off rate. (c-d)
show corresponding results for $n=1000$, $pq=10^{-4}$ and a
target block error of $10^{-3}$. 
}
\vspace{-0.2in}
\label{fig:5}
\end{figure*}

While prefactors are not captured in the asymptotic analysis of
Theorems \ref{prop:channel_coding_theorem} and \ref{prop:IMLE},
they allow the following approximations. The conditional
ML probability of error is approximately $2^{-n\error(R)} 
        \text{ for } R <1-q\h2(p)$,
which holds true regardless of whether it is identified by SRGRAND or brute force,
where the expression for $\error(R)$ can be found in Eq. \eqref{eq:error}. 
For SRGRAND decoding, our measure of complexity is the average number of
guesses per bit per decoding, approximately
$2^{n\min(1-R,\LambdaL(\Hhalf))}/n.$
For comparison, we define the complexity of the computation
of the ML decoding in Eq. \eqref{eq:straight_MLE} 
by the method described
in \cite{BCJR74}
to be
the number of conditional probabilities that must be computed
per bit before rank ordering and determining the most likely codebook
element, equal to 
$2^{n\min(R,1-R)}/n$,
where we are equating the work performed in one noise guess, which
amounts to checking if a string is an element of the codebook,
with the computation of one conditional probability.

For two values of block size, $n=100$ and $n=1000$, and $(q,p)$
pairs such that $pq$ is constant and so comparable with the hard
detection channel, Fig.~\ref{fig:5} plots the approximate error
probabilities and complexity as a function of codebook rate. The
upper panels show the error probabilities with a target block
error rate indicated by the dashed horizontal line. 
The provision of symbol reliability information greatly improves the block
error probability, even though in this comparison the conditional probability
of a bit flip given symbol reliability information increases as the symbol reliability 
probability decreases.

The lower two panels show the approximate complexity. The dashed
line gives the approximate complexity for the approach in \cite{BCJR74}, which grows exponentially in $R$,
when computing a conditional probability for every codeword. In
contrast, the complexity of the SRGRAND approach is initially flat. As the rate, $R$, increases, eventually the
SRGRAND complexity drops, as encountering an erroneous element of the
codebook clips the long guessing tail of true noise.  The diamonds
indicate the rate above which the target block error rate would be
violated. 

\section{Empirical Performance Evaluation} 
\label{sec:perf}
\begin{figure*}
\begin{center}
\vspace{-0.2in}
\includegraphics[width=0.7\textwidth]{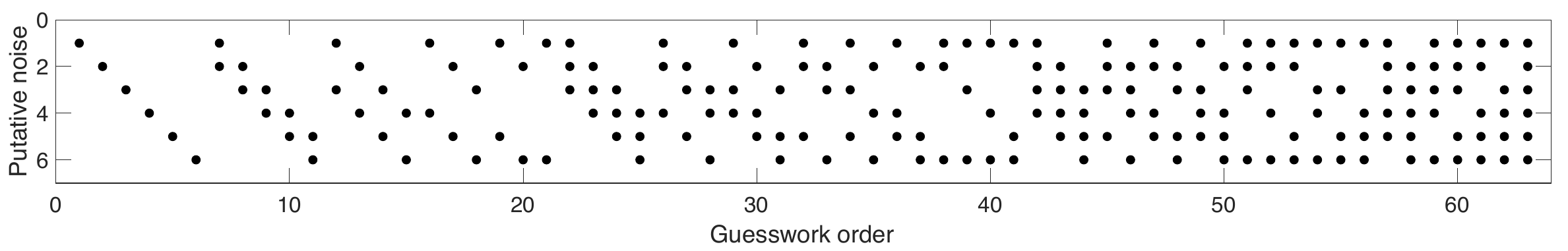}\\
\includegraphics[width=0.35\textwidth]{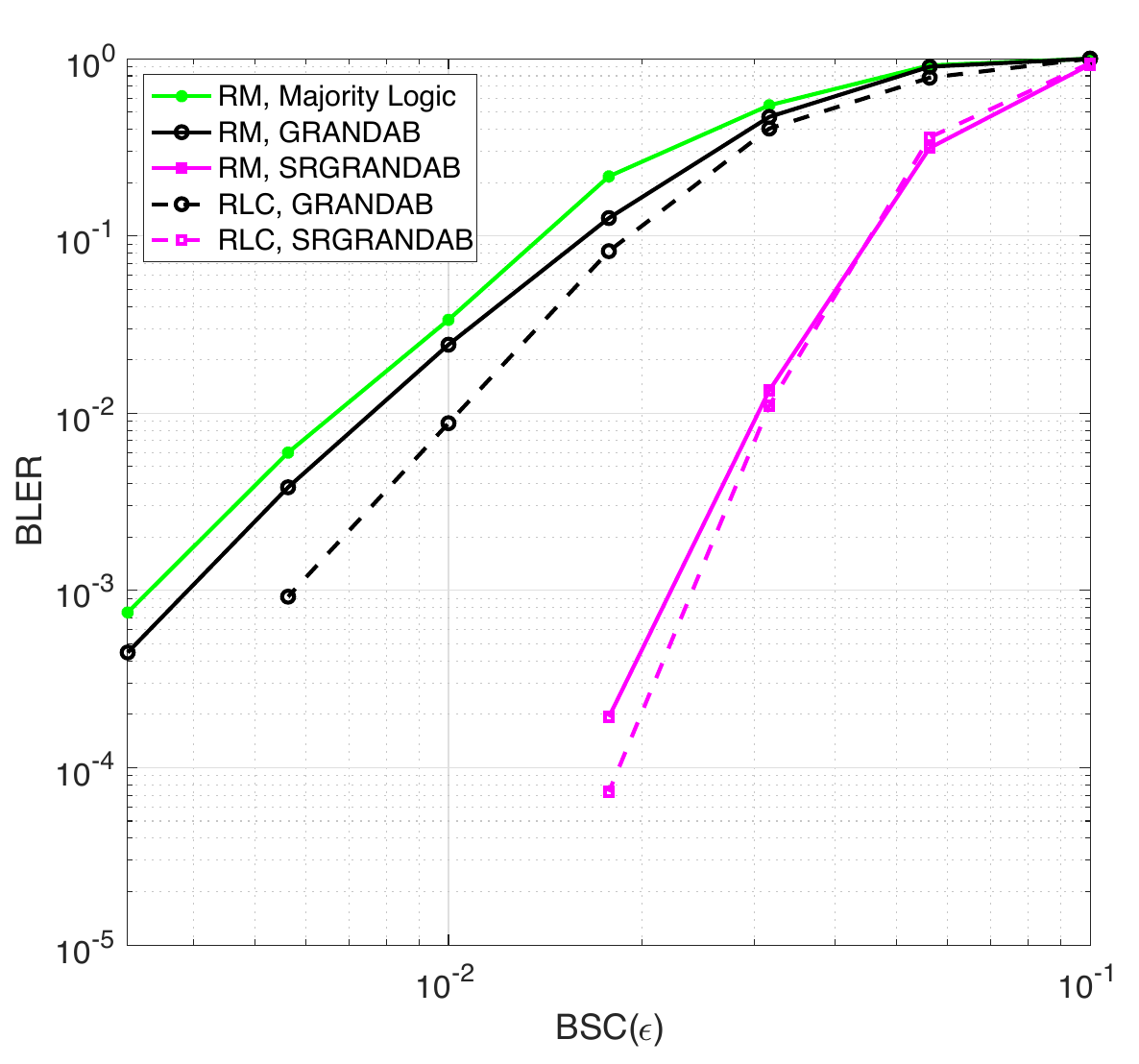}
\includegraphics[width=0.35\textwidth]{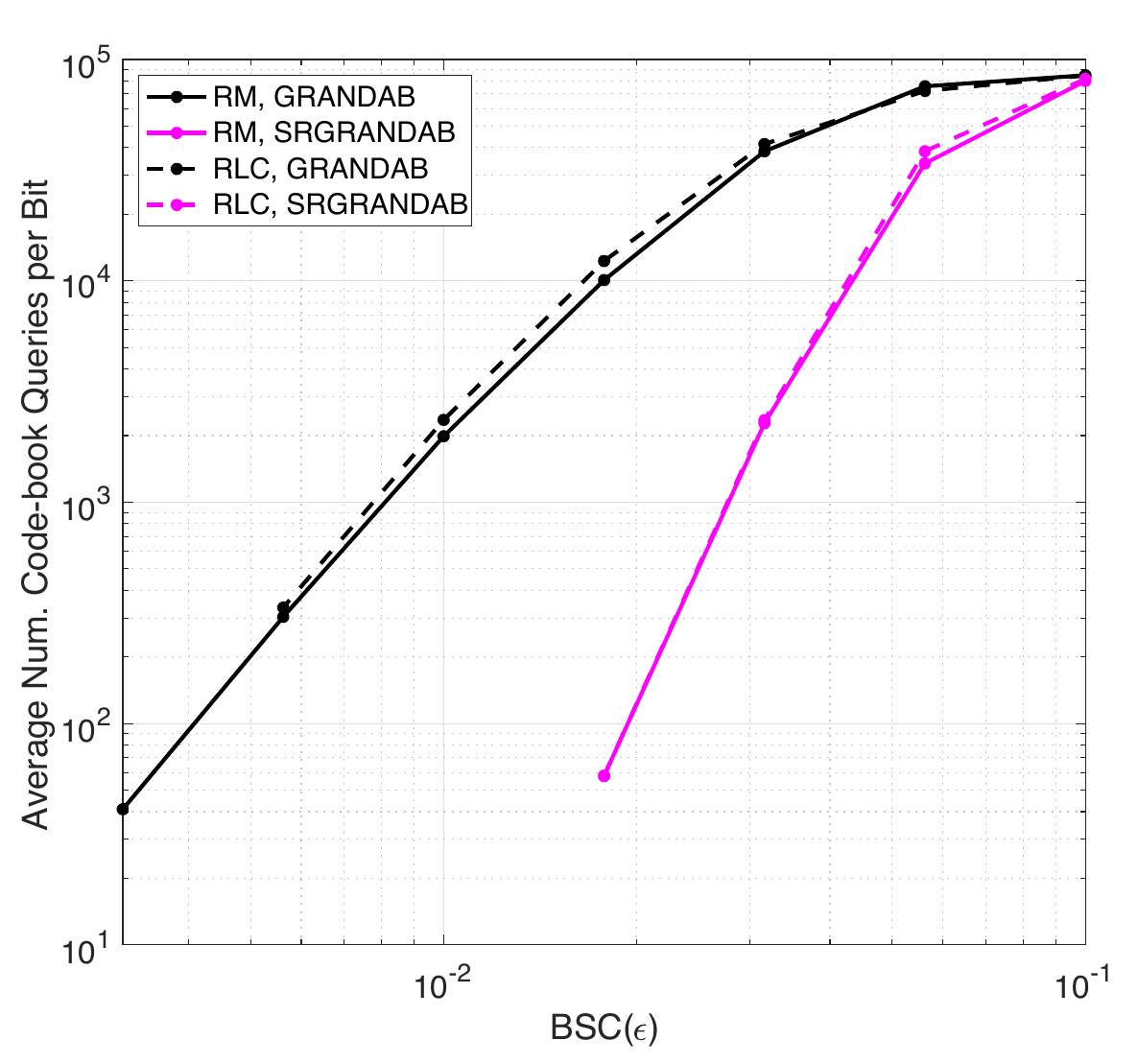}
\vspace{-0.2in}
\end{center}
\caption{
Performance evaluation with RM and RLC $[128,99]$, rate
$0.77$, codes
in an SR-BSC. Upper panel illustrates GRAND and SRGRAND guessing
order on an $n=7$ code, where each column is a putative noise
sequence with a dot indicating a $1$, and sequences are queried in
order from left to right. Left panel gives BLER performance
for majority logic decoding, GRANDAB and SRGRANDAB decoding of a
RM code as well as GRANDAB and SRGRANDAB decoding of RLCs. Lower
right panel gives average number of codebook queries per bit per
decoding for GRANDAB and SRGRANDAB.\vspace{-0.2in}
}
\label{fig:6}
\end{figure*}

A distinctive aspect of the GRAND approach is that it is readily implemented
and can be used
with any block code construction. While the theoretical results
 in Section \ref{sec:SR-BSC} are for uniform-at-random
codebooks, in practice nearly all error correction codebooks are
 linear in a finite field with $k$ input bits and $n$
coded bits.  Associated with each code is a check matrix $H^{n\times
n-k}\in\{0,1\}^{n\times n-k}$ and to test if a string, $y^n$, is
in the codebook a single matrix multiplication and comparison,
$H^{n\times n-k}(y^n)^T\MYEQ(0^{n-k})^T$, suffice, in the appropriate
field. Here we compare the decoding performance of GRANDAB, SRGRANDAB
for four types of binary linear codes.

\subsection{The Symbol Reliability Binary Symmetric Channel}

In the context of the SR-BSC introduced in Section \ref{sec:SR-BSC}, when the unconditional bit flip
probability is $\epsilon$, we set the probability that a bit is
marked as unreliable to be $q=\sqrt{\epsilon}$ and the bit flip
probability conditioned on unreliability to be $p=\sqrt{\epsilon}$.
For hard detection GRANDAB, putative noise strings are queried in
order of Hamming weight. Within each set of strings with the same
Hamming weight, the ordering is arbitrary and we do so in
the order illustrated in Fig.~\ref{fig:6}, first panel. For SRGRANDAB,
we assume that the channel state is known and use the same search
pattern, but confined to querying only bits for which the channel
state was marked as unreliable for any given communication. For
GRANDAB, we set the abandonment threshold to check for up to four
bit-flips. For SRGRANDAB, we allow the same number of codebook
queries as GRANDAB before abandoning and reporting a decoding error.

RM codes, which only exist for some $[n,k]$ pairs,
are broadly used in wireless communications and have a well-established
hard detection decoder, majority logic decoding \cite{shu2004}.
Fig.~\ref{fig:6} reports Block Error Rates (BLER) as a function of
the bit flip probability $\epsilon$ for a rate $0.77$, $[128,99]$, RM
code. As majority logic
decoding is tailored to a BSC and is known to be accurate in that
setting, its performance is only slightly degraded from the ML BLER
that GRANDAB provides. The provision of reliability information to
SRGRANDAB gives it a distinct advantage, resulting in significantly
enhanced BLER. The right panel reports the average number of codebook
queries per received bit that GRANDAB and SRGRANDAB make. As each query solely
requires a matrix multiplication by a sparse vector, for typical
target BLER of $10^{-2}$ or lower,
the complexity requirements of GRANDAB and SRGRANDAB are modest.

Since the 1960s, RLCs have been known to be capacity-achieving if
twined with ML decoding \cite{Gal68} with the same error exponents 
as those for uniform-at-random codebooks \cite{DZF16}.
Those results hinge on a proof that at high rates the average RLC is a good one. The lack of an efficient
decoder that can accurately decode any linear, high-rate codebook
has meant, however, that this avenue is little explored. Here we
consider the application of GRANDAB and SRGRANDAB for decoding RLCs. For any $[n,k]$ pair we can construct systematic binary RLCs by making a
random generator matrix $\left[I^{k\times k} | \LC^{k\times
n-k}\right]$, where $I^{k\times k}$ is the identity matrix and the
entries of the random check matrix $\LC^{k\times n-k}$ are independent
Bernoulli 1/2 random variables. To check if $y^n$ is a member of
the codebook, one can test if $y^k \LC^{k\times n-k}\MYEQ
(y_{k+1},\ldots,y_n)$, obviating the need for the receiver to
determine the associated check matrix. Consistently with theoretical
results, in an empirical evaluation codes are re-randomized after
each use. In practice,  the sender and
receiver could share a seed for the random number generator from which
the check matrix is produced.

Fig.~\ref{fig:6} also reports the BLER and complexity performance of
GRANDAB and SRGRANDAB for [128,99] RLCs, so that the results are
directly comparable to those for RM codes. With hard detection ML
decoding by GRANDAB, it can be seen that RLCs slightly outperform
RM codes, leading to better BLER and comparable decoding complexity.
This result is potentially surprising as the re-randomization in
the RLC would lead one to suspect that some codes are poor performers,
but is consistent with theory that says that RLCs are typically good. The provision of symbol reliability
information changes matters and SRGRANDAB gets equally good performance
from both RM and RLCs. The use of RLCs, which necessitates 
a universal decoder, holds appeal as changing codebooks
may provide enhanced security, and our results suggest there is
no loss in terms of error performance in using them.

\begin{figure*}
\begin{center}
\vspace{-0.2in}
\includegraphics[width=0.35\textwidth]{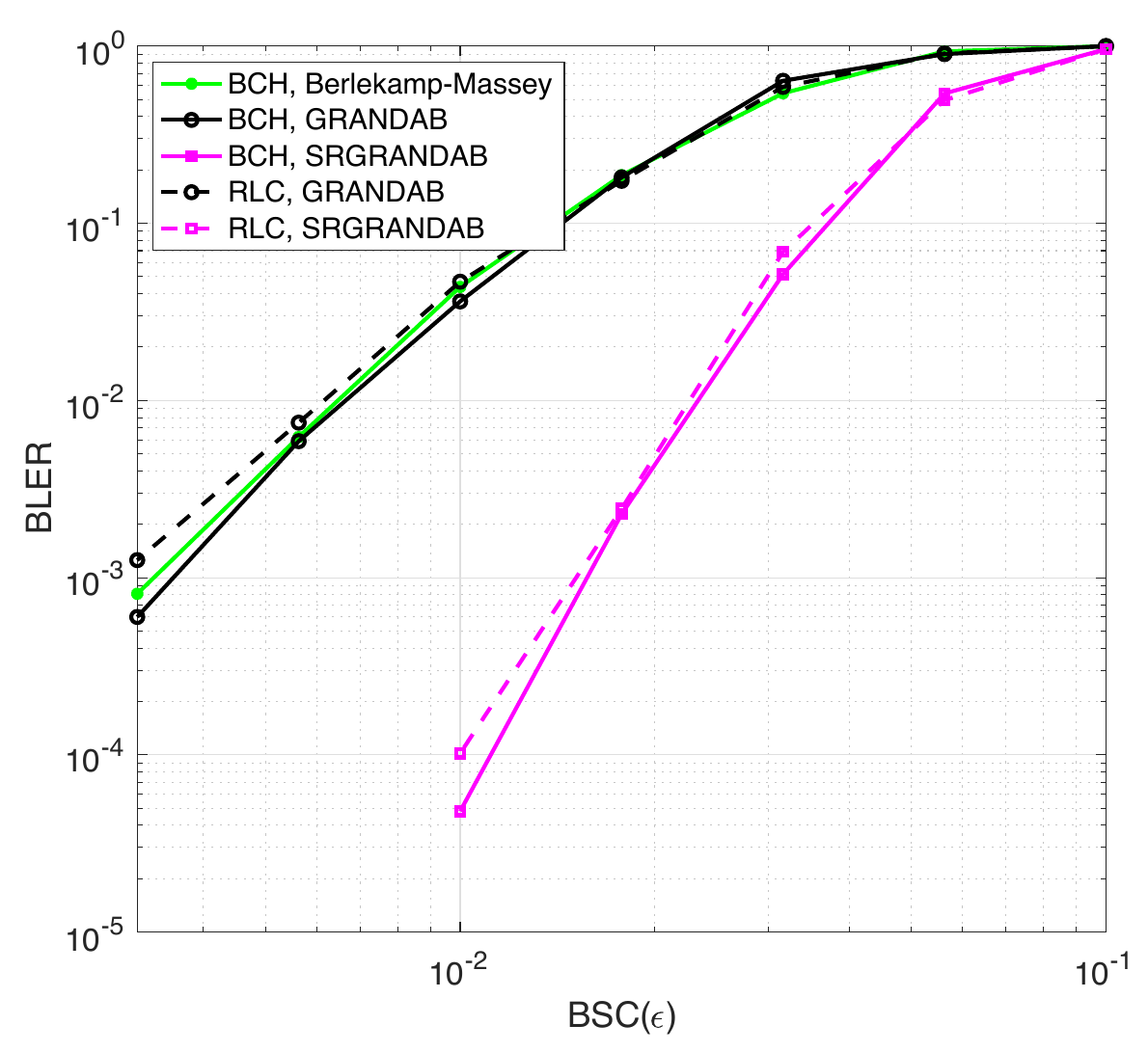}
\includegraphics[width=0.35\textwidth]{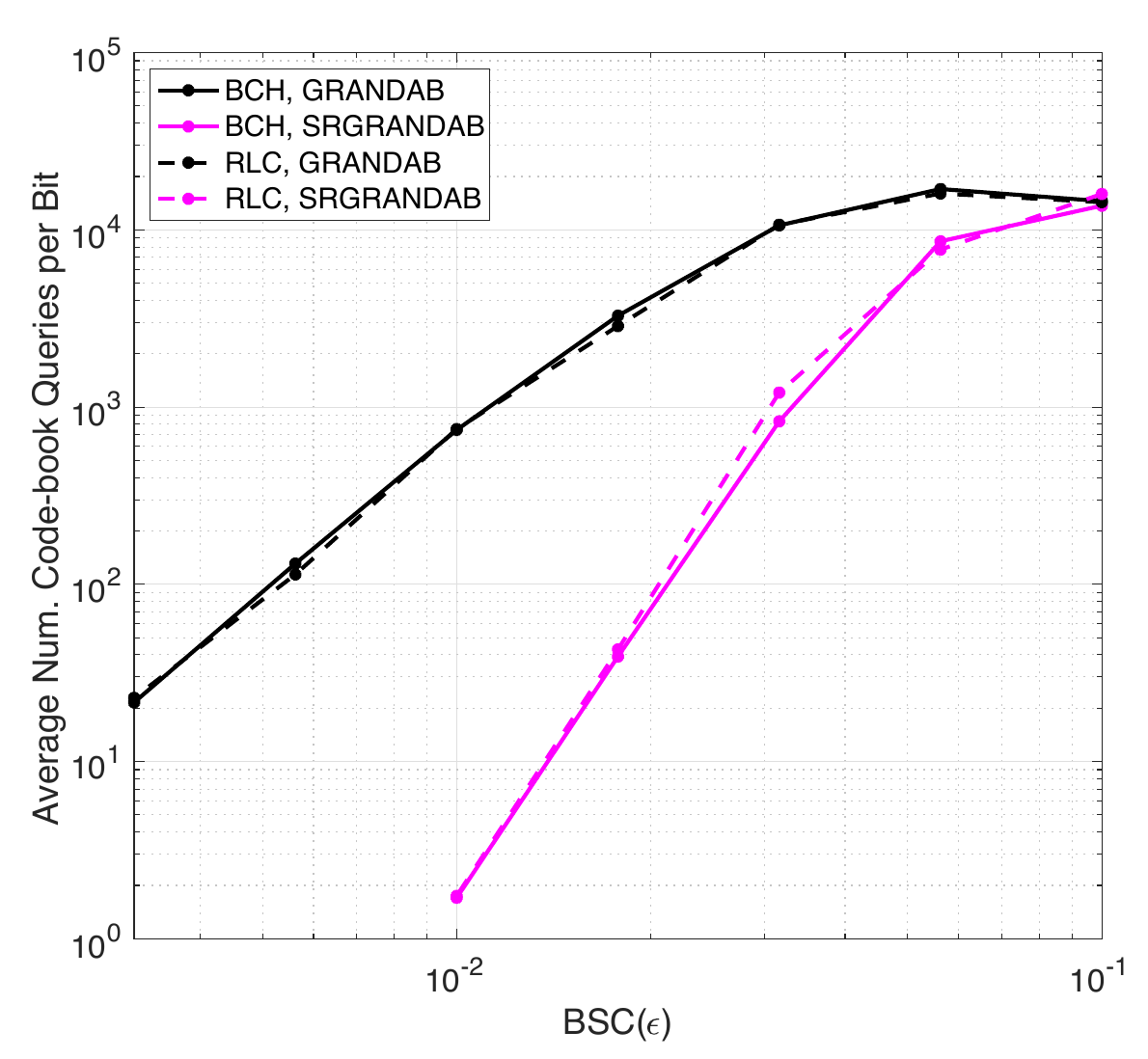}
\vspace{-0.2in}
\end{center}
\caption{
Performance with BCH and RLC $[127,106]$, rate $0.83$,
codes in an SR-BSC. Guessing order as in the upper panel of
Fig.~\ref{fig:6}. Left panel gives BLER performance for BM, GRANDAB and SRGRANDAB decoding of a BCH code, as well as
GRANDAB and SRGRANDAB decoding of RLCs. Right panel gives average
number of codebook queries per bit per decoding for GRANDAB and
SRGRANDAB.
\vspace{-0.2in}}
\label{fig:7}
\end{figure*}

For a rate $0.83$,
BCH $[127,106]$ Fig.~\ref{fig:7} reports BLER as a function of
$\epsilon$, as well as RLCs of the same rate. The results mirror those
found for RM codes, where the dedicated hard detection decoder
provides similar performance to the universal GRANDAB and the provision
of symbol reliability information leads SRGRANDAB to significantly outperform
both. As with RM codes, RLCs, which can only be efficiently decoded with
the GRAND approach, lead to similar BLERs as the BCH code
with essentially identical complexity for both. The latter is not
surprising as the complexity of GRANDAB and SRGRANDAB is largely dominated
by properties of the noise rather than those of the codebook. These results suggest that for BLER performance of moderate-redundancy codes,
the accuracy of the decoding mechanism is more important than the
codebook structure, opening up a rich palette of code sizes and
rates for URLLC in a single algorithmic instantiation.

\subsection{Quantizing Soft Information to Create Symbol Reliability Information}

The mathematical analysis assumes that the mask provided to the decoder, $s^n$, is correct, with symbols accurately tagged as reliable or unreliable. In practice, that requires binary quantization of soft information. 
Should quantization result in a symbol being reliable when it is not, that would necessarily result in an erroneous decoding or abandonment. Here we illustrate simple means by which mask creation can be achieved such that the frequency of provision of erroneous masks, the Mask Error Rate ($\MERR$), does not dominate the BLER. The masking rule is a function of the SNR, the length and redundancy of the code. 

Consider an AWGN with noise variance is $\sigma^2$ and a transmitter-receiver pair employing BPSK with transmitted the binary symbols corresponding to $\pm1$. We wish to identify a threshold, $\tau$, such that if the absolute value of a received signal is beyond $\tau$ it is likely to be reliable. Given $\tau$, the probability an individual bit
is erroneously labeled as reliable when it is incorrect is
$P(\sigma\NORMAL>1+\tau)$, where $\NORMAL$ is a Gaussian with mean zero and variance one. Thus, for a code of length $n$ we have that the likelihood one or more bits are erroneously marked as reliable, resulting in a mask error, is
$\MERR = 1-P(\sigma\NORMAL\leq 1+\tau)^n$. Hence, by setting a target $\MERR$ as a function of the code length and SNR such that the $\MERR$ will not dominate the BLER, the receiver determines the static signal threshold by $\tau = \sigma \Ninv((1-\MERR)^{1/n})-1$, where $\Ninv$ is the inverse of a Normal distribution. In addition, we set the mask so it always includes the $n-k$ least reliable bits, which allows SRGRAND to do a small amount of corrective work if necessary. In particular, that avoids circumstances where the reception is indicated to be error free, but the received demodulated signal is not in the codebook.

\label{sec:perf_awgn}
\begin{figure*}
\begin{center}
\includegraphics[width=0.32\textwidth]{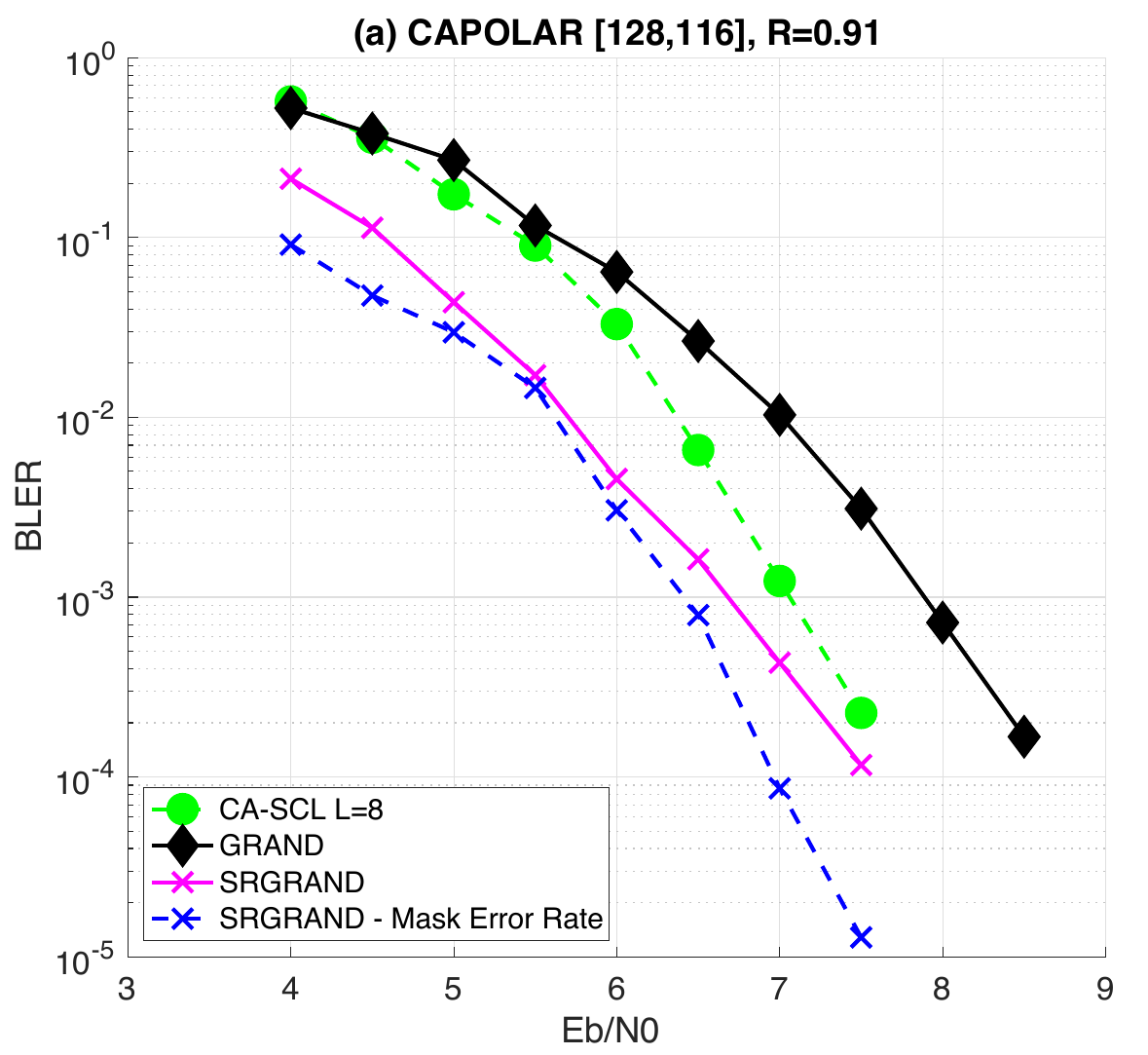}
\includegraphics[width=0.32\textwidth]{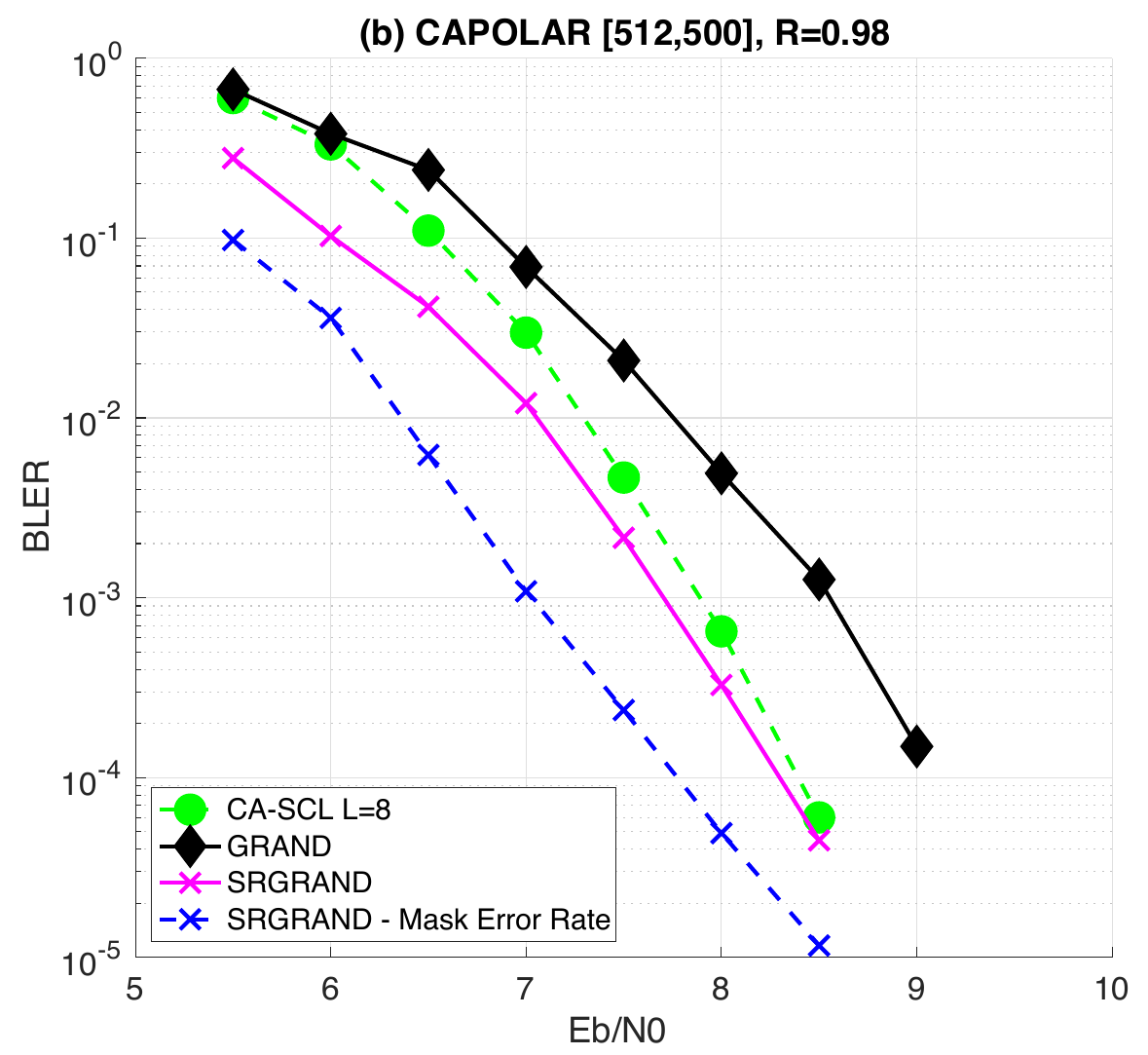}
\includegraphics[width=0.32\textwidth]{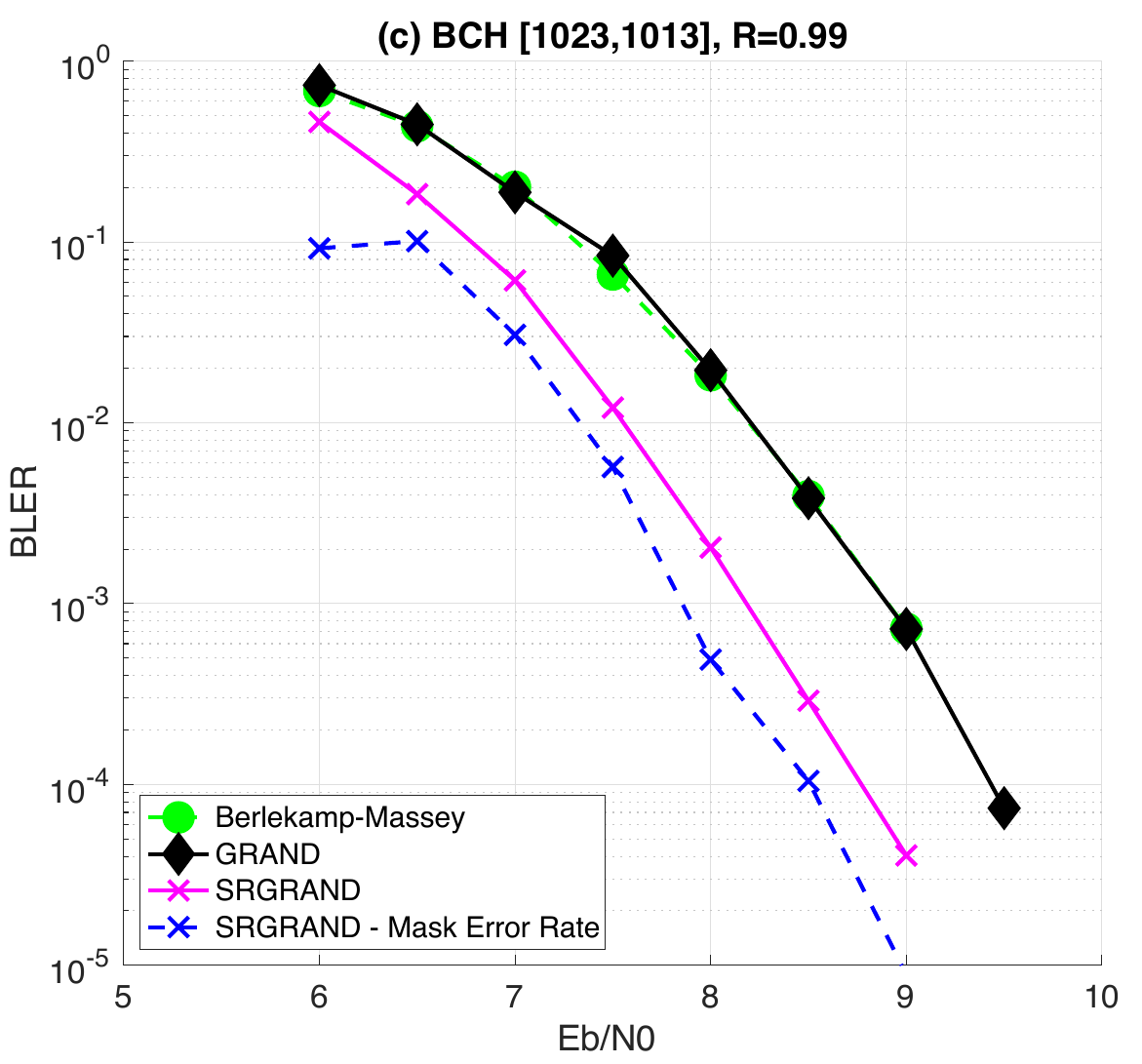}
\vspace{-0.2in}
\end{center}
\caption{
Performance evaluation in an AWGN with a threshold mask for CA-Polar [128,116], CA-Polar [512,500], and BCH [1023,1013] codes. Green lines correspond to full soft detection CA-SCL decoding of CA-Polar codes with a list size of eight, and hard detection BM decoding for the BCH code. Black lines correspond to ML GRANDAB decoding, magenta to SRGRAND, with blue lines indicating the contribution to SGRAND's BLER that comes from mask errors due to erroneous quantization of soft information.  \vspace{-0.2in}
}
\label{fig:8}
\end{figure*}

Figure \ref{fig:8} shows results of this system for two CA-Polar codes, [128,116] and [512,500], and a BCH [1023,1013] code.  Green lines correspond to full soft decoding with CA-SCL of CA-Polar codes \cite{Cassagne2019a} using a list-size of $8$ \cite{Zhaetal18, Chenetal19, Kesteletal20}, and BM decoding of the BCH code. Black lines indicate optimal ML decoding with hard detection GRAND. The blue lines indicate the contribution to BLER that comes from the masks being in error, $\MERR$, where static target mask error rates are determined in advance and used to identify the marking threshold $\tau$. The magenta lines report the overall SRGRAND BLER, inclusive of the $\MERR$. At a BLER of $10^{-3}$, these results demonstrate a BLER gain of 0.75 to 1dB can be obtained with SRGRAND over ML hard detection decoding, irrespective of code-length.

\section{Discussion}
\label{sec:discussion}

We have introduced SRGRAND and SRGRANDAB,   two noise-centric decoding algorithms using symbol reliability information.
By using the symbol reliability information to mask
symbols that are reliable and guessing noise only on unreliable
symbols,  these algorithms can realize higher rates, with lower error
probabilities, and less complexity, than without symbol reliability information.

All of the GRAND algorithms are suitable for use with any codebook
so long as testing membership of the codebook for a string of symbols
is efficient. For linear codes, such testing requires only a matrix multiplication over a finite
field. CRC codes, CA-Polar Codes and RLCs are all
linear.
Moreover,
guesswork orders are known to be robust to mismatch \cite{Sundaresan07},
and so decoding precision should not be sensitive to minor imprecision
in the channel noise model.

We empirically 
compared SRGRAND and SRGRANDAB with the well established majority logic decoding of RM
codes and BM decoding of a BCH code.  The provision of symbol reliability information
to SRGRANDAB results in substantially better performance.
As the algorithms are universal, they enable us to empirically
consider decoding RLCs, which is little explored outside of theory.
The BLER
performance is comparable with the highly structured RM and BCH codes of the same rate.
This opens the possibility of using SRGRANDAB for 
security, based on a principle of having the sender and
receiver use a distinct linear code drawn using a
cryptographically secure random number generator for each transmission.

While we presented results for one SRGRANDAB abandonment rule that
reduces average algorithmic complexity without sacrificing channel
capacity, others are possible and, indeed, can be used in combination.
Here we mention two more. The first is a natural extension to the
rule of abandoning guessing when coverage of the typical set for
the average number of potentially noise impacted symbols. In the
symbol reliability model, the specific number of potentially noise-impacted
symbols, $L^n$, for each received transmission, $Y^n$, is known to
the algorithm and querying is abandoned after $|\A|^{L^n(\HN+\delta)}$
guesses.
Analysis of the impact of this rule on error exponents and complexity
follows the same line of argument as presented in the paper, though
the resulting expressions are less elegant. A distinct alternative
is not to guess at all if too many symbols are reported to be
potentially noise impacted; i.e. if $L^n>n(\muL+\delta)$. It is
straight forward to show this rule does not impact capacity, but
an analysis of complexity, which would now be conditional on $L^n\leq
n(\muL+\delta)$, would not follow immediately from the large deviation
arguments presented here. The analysis in this paper for codes of
fixed length could, however, be readily extended to decoding with
symbol reliability information for variable length codes \cite{WV02, Wei03} and
rateless codes.

SRGRAND avails of symbol reliability information, which is
the most succinct form of soft information, and lends itself to both mathematical
analysis and implementation in hardware, seeing a $0.5$ to $0.75$dB gain 
over hard detection GRAND. A natural question is how to use more fine-grained
soft information in a GRAND algorithm, what the additional algorithmic
complexity would be, and what performance gains would be available.
By creating a bespoke noise effect query order for each reception, it is possible 
to use one real-valued piece of soft information per
bit to identify soft-detection ML decodings \cite{solomon20}. Although
the resulting algorithm does not lend itself to theoretical determination of
error exponents or to efficient implementation in hardware due to the need for
dynamic memory, a software implementation enables
the empirical evaluation of a bound on the achievable performance for a given
code. A heuristic algorithm that uses $\log_2(n)$ bits of soft information 
per received bit has been reported that appears to empirically approximate the 
performance available from full soft information with a simpler algorithm \cite{duffy21,An21}.
Again, its construction does not lend itself to mathematical identification of error exponents, 
but it is more suitable for implementation in hardware and an
architecture for it has been proposed \cite{abbas2021}, albeit one that is significantly more complex in terms
of energy and area than is the case for GRAND \cite{abbas2020high,Riaz21} or would be for SRGRAND. 
The question of whether SRGRAND could be augmented to avail of more finely quantized soft information while retaining
the simplicity of its operation remains outstanding.

The GRAND algorithms can themselves provide, in addition to
a decoding, soft information through the number of noise queries. A
lower number of guesses corresponds to a higher likelihood of correct
decoding. Such soft information can be of use, for example, for
component codes in a concatenated code or Turbo code \cite{HH89,
BAAF93,FBLH98,WH00}. Thus one may envisage using the information
on decoding reliability of SRGRAND and SRGRANDAB in a manner akin to
the reliability information provided by the Soft-Output Viterbi
Algorithm \cite{HH89,BAAF93, PRV96, FBLH98, GS00}, by the operation
of Turbo decoding \cite{WH00, CFR00, HP94,KImetal10, WP99, SSS04, johannesson2015fundamentals},
by the syndrome information used in Ordered Statistics Decoding
(OSD) \cite{FL95, KI02, MVFT03}, or other soft-input, soft-output
schemes \cite{MG11, FKU16, Wei03, SB10}. In general, we can envisage
in future work systems that meld equalization and decoding as in
\cite{Ung74} or soft information originating from other
decoding processes, \cite{For65,BM97, BDMP97, SM01, BKMP10}.

\vspace{-0,1in}
\bibliographystyle{IEEEtran}

\bibliography{SRGRAND}

\end{document}